\documentclass[11pt]{article}
\usepackage{amsfonts}
\usepackage{amsmath}
\usepackage{amssymb}
\usepackage{graphicx}
\usepackage{algorithm}
\usepackage{algorithmic}
\usepackage{cases}
\usepackage{epstopdf}
\usepackage{geometry, tabularx}
\usepackage{multirow,makecell}
\usepackage{enumerate}
\usepackage{bm}
\usepackage{pifont}
\usepackage{footnote}
\usepackage{titlesec}
\usepackage[titletoc]{appendix}
\usepackage{setspace}
\usepackage[sectionbib]{natbib}

\usepackage{mathrsfs}
\usepackage{threeparttable}

\providecommand{\U}[1]{\protect\rule{.1in}{.1in}}

\newtheorem{thm}{\bf Theorem}      
\newtheorem{lem}{\bf Lemma}
\newtheorem{prop}{\bf Proposition}

\newtheorem{rem}{\bf Remark}

\usepackage{hyperref}
\hypersetup{hypertex=true,
            colorlinks=true,
            linkcolor=blue,
            anchorcolor=green,
            citecolor=blue}

\usepackage{graphicx,color}
\usepackage[bf,SL,BF]{subfigure}

\newenvironment{proof}[1][Proof]{\noindent\textbf{#1.} }{\ \rule{0.5em}{0.5em}}
\normalsize
\titlespacing*{\section} {0pt}{9pt}{0pt}
\numberwithin{equation}{section}
\allowdisplaybreaks

\usepackage{booktabs}
\usepackage{diagbox}
\begin{document}
\normalsize
\title{Parallel ADMM Algorithm with Gaussian Back Substitution for High-Dimensional Quantile Regression and Classification}

\author{ Xiaofei Wu\thanks{College of Mathematics and Statistics, Chongqing University, Chongqing, 401331, P.R. China. Email: xfwu1016@163.com}, \ \ 
        Dingzi Guo \thanks{Guo and Wu have contributed equally to this paper. Institute of Financial Studies, Shandong University, Shandong, 250100, P.R. China. Email: guodingzi@mail.sdu.edu.cn.},\\  
        Rongmei Liang \thanks{Department of Statistics and Data Science, Southern University of Science and Technology, Guangdong, 518055, P.R. China. Email: liang\_r\_m@163.com},\\ 
       Zhimin Zhang\thanks{Corresponding author. College of Mathematics and Statistics, Chongqing University, Chongqing, 401331, P.R. China. Email: zmzhang@cqu.edu.cn}.
}
\date{}
\maketitle

\begin{abstract}
In the field of high-dimensional data analysis, modeling methods based on quantile loss function are highly regarded due to their ability to provide a comprehensive statistical perspective and effective handling of heterogeneous data. In recent years, many studies have focused on using the parallel alternating direction method of multipliers (P-ADMM) to solve high-dimensional quantile regression and classification problems. One efficient strategy is to reformulate the quantile loss function by introducing slack variables. However, this reformulation introduces a theoretical challenge: even when the regularization term is convex, the convergence of the algorithm cannot be guaranteed. To address this challenge, this paper proposes the Gaussian Back-Substitution strategy, which requires only a simple and effective correction step that can be easily integrated into existing parallel algorithm frameworks, achieving a linear convergence rate. Furthermore, this paper extends the parallel algorithm to handle some novel quantile loss classification models. Numerical simulations demonstrate that the proposed modified P-ADMM algorithm exhibits excellent performance in terms of reliability and efficiency.
\end{abstract}

\section{Introduction}
Quantile regression, pioneered by \cite{Koenker1978Regressions}, explores how a response variable depends on a set of predictors by modeling the conditional quantile as a function of these predictors. Unlike mean regression, which focuses solely on estimating the conditional mean of the response, quantile regression offers a more precise representation of the relationship between the response and the predictors.  Furthermore, quantile regression exhibits superior robustness when handling datasets with heterogeneous characteristics and can effectively process data with heavy-tailed distributions, owing to its less restrictive assumptions regarding error distribution.  Quantile loss is also utilized in support vector machines (SVM) for classification purposes (see  Section 9.3 in \cite{Christmann2008Support} and Proposition 1 in \cite{Wu2025Multi}). Compared to traditional SVM (support vector machine) in \cite{Vapnik1995The}, quantile loss SVM has been shown to be less sensitive to noise around the separating hyperplane, making it more robust to resampling. For a detailed discussion, see \cite{Huang2014Support}.

Consider a regression or classification  problem with \(n\) observations of the form 
\begin{equation}\label{data}
\{ {y_i}, \boldsymbol{x}_i \}_{i=1}^n = \{ {y_i}, {x_{i,1}}, {x_{i,2}}, \ldots, {x_{i,p}} \}_{i=1}^n = \{\boldsymbol{y}, \boldsymbol{X}\} = \boldsymbol{D},
\end{equation}
where the data is assumed to be a random sample from an unknown joint distribution with a probability density function. The random variable \(y\) represents the ``response" or ``outcome", while \(\boldsymbol{x} = \{x_1, x_2, \dots, x_p\}\) denotes the predictor variables (features). These features may include the original observations and/or functions derived from them. If considering the intercept term, the first column of $\bm X$ is set to all 1. Without loss of generality, ${y}_i$ is quantitative for regression models, but equals to -1 or 1 for classification models.  For a given quantile $\tau \in (0,1)$, one can obtain an estimate of quantile regression by optimizing the following objective function,
\begin{equation*}
\hat{\bm \beta}(\tau) = \arg \min_{\bm \beta} \sum_{i=1}^{n} \rho_{\tau} \left( y_i - \boldsymbol{x}_i^\top \bm \beta \right),
\end{equation*}
where \(\rho_{\tau}(u) = u [\tau - I(u < 0)]\) (for \(u \in \mathbb{R}\)) is the check loss, and \(I(\cdot)\) is the indicator function.  Note that when $\tau = 1$, the quantile loss degenerates into hinge loss in SVM for classification. Many studies have shown that quantile loss  can also be used for classification when $\tau \in (0,1)$, such as \cite{Christmann2008Support}, \cite{Huang2014Support}, \cite{Liang2024Linearized} and \cite{Wu2025Multi}.
The estimator for quantile loss SVM can be obtained by solving the following optimization problem,
\begin{align*}
\arg \min_{ \bm \beta} \left\{ \sum_{i=1}^{n} \rho_{\tau} \left( 1 - y_i \boldsymbol{x}_i^\top \bm \beta  \right) +  \| \bm \lambda \odot \bm \beta\|_2^2 \right\},
\end{align*}
where $\bm \lambda$ is a penalization parameter vector, where the first element is always set to 0, and the remaining elements are constrained to be non-negative. Here, the symbol $\odot$ is commonly used to denote the Hadamard product, also known as the element-wise product or the entrywise product.   Clearly, when $\tau = 1$, $\rho_{\tau}$ degenerates into the hinge loss commonly used in SVM, and the above expression simplifies to the form of the ordinary SVM, specifically, hinge loss plus ridge penalty term.

To accommodate high-dimensional scenarios where \(p > n\), it is common practice to substitute the ridge term \(\| \bm \lambda \odot \bm \beta\|_2^2 \) (\cite{Hoerl1970Ridge}) with sparse regularization techniques, including LASSO in \cite{Tibshirani1996Regression}, elastic net in \cite{Zou2006The}, adaptive LASSO in \cite{Zou2006The}, SCAD in \cite{Fan2001Variable}, and MCP in \cite{Zhang2010Nearly}.   Consequently, the sparse penalized quantile  regression and classification formulation becomes
\begin{align}\label{pqr}
\hat{\bm \beta}(\tau) = \arg \min_{\bm \beta} \left\{ \sum_{i=1}^{n} \rho_{\tau} \left( \tilde{y_i} - \tilde{\boldsymbol{x}}_i^\top \bm \beta \right) + P_{\bm \lambda} (|\bm \beta|) \right\}.
\end{align}
Here, \(P_{\bm \lambda}(|\bm \beta|)\) denotes a sparse regularization term that is separable, meaning:
\[
P_{\bm \lambda}(|\bm \beta|) = \sum_{j=1}^{p} P_{\lambda_j}(|\beta_j|).
\]
This formulation allows for a distinct penalization of each component \(\beta_j\) using its respective regularization parameter \(\lambda_j\).  For  regression,   $\tilde{y}_i$ is $y_i$, $\tilde{\boldsymbol{x}}_i$ is ${\boldsymbol{x}}_i$;    for classification, $\tilde{y}_i$ is 1;  $\tilde{\boldsymbol{x}}_i$ is $y_i{\boldsymbol{x}}_i$.  Optimization formula (\ref{pqr}) is a highly flexible expression hat can represent numerous quantile regression and SVM classification models, including penalized quantile regression (\cite{Belloni2011L1}, \cite{Wang2012Quantile}, \cite{Fan2012Adaptive} and \cite{Gu2018ADMM}) and SVM with sparse regularizationin (\cite{Zhu20031-norm}, \cite{Wang2006The}, \cite{Zhang2016Variable} and \cite{Liang2024Linearized})

The advancement of modern science and technology has made data collection increasingly effortless, leading to an explosion of variables and vast amounts of data. Due to the sheer volume of data and other factors such as privacy concerns, it has become essential to store it in a distributed manner. Consequently, designing parallel algorithms that can effectively manage these large and distributed datasets is crucial.  
Several parallel algorithms have been proposed to address problem \eqref{pqr}, including QR-ADMM (ADMM for quantile regression) in \cite{Yu2017ADMM}, QPADM (quantile regression with parallel ADMM) in \cite{Yu2017A} and \cite{Wu2025A}. More recently, inspired by \cite{Guan2018An}, \cite{Fan2021Penalized} introduced a slack variable representation of the quantile regression problem, demonstrating that this new formulation is significantly faster than QPADM, especially when the data volume \(n\) or the dimensionality \(p\) is large. In addition, the slack variable representation is also used by \cite{Guan2020An} to design parallel algorithms for solving SVMs with sparse regularization.

However,  these slack variable representations raise a new issue, which is that convergence cannot be proven mathematically. The main reason for this is that the slack variable representation introduces two slack variables, which transform the parallel ADMM algorithms (both distributed and non-distributed) into a three-block ADMM algorithm, see Section \ref{sec31} for  detailed information. \cite{Chen2016The}  demonstrated that directly extending the ADMM algorithm for convex optimization with three or more separable blocks may not guarantee convergence, and they even provided an example of divergence. Therefore, although the parallel ADMM algorithms  proposed by \cite{Guan2020An} and \cite{Fan2021Penalized} did not  exhibit non-convergence of iterative solutions in numerical experiments, there is no theoretical guarantee of convergence for the iterative solutions, even when the optimization objective is convex.

In this paper, we apply the Gaussian back substitution technique to refine the iterative steps, which allows the parallel ADMM algorithm proposed by \cite{Guan2020An} and \cite{Fan2021Penalized}  to achieve a linear convergence rate. This Gaussian back substitution technique is straightforward and easy to implement, requiring only a linear operation on a portion of the iterative sequences generated by their algorithm.
Besides demonstrating that the algorithm in   \cite{Guan2020An} and \cite{Fan2021Penalized} can theoretically guarantee convergence with a simple adjustment, the main contributions of this paper are as follows:
\begin{enumerate}
\item We suggest changing the order of variable iteration in \cite{Fan2021Penalized} such that our Gaussian back substitution technique involves only simple vector additions and subtractions, thereby eliminating the need for matrix-vector multiplication. 
 Although this change may seem minor, it can significantly enhance computational efficiency in algorithms where both \( n \) and \( p \) are relatively large. More importantly, this change will not impact the linear convergence of the algorithm.

\item  This paper proposes some new classification models with nonconvex regularization terms based on quantile loss. Leveraging the equivalence of quantile loss in classification and regression tasks, it indicates that existing parallel ADMM algorithms for solving penalized quantile regression, as well as those proposed in this paper, can also be applied to solve these new classification models.
\end{enumerate}

The remainder of this paper is organized as follows. Section \ref{sec2} provides a review of relevant literature along with an introduction to preliminary knowledge. Section \ref{sec3} outlines the existing QPADM-slack algorithm, incorporating adjustments to the Gaussian back-substitution process to achieve linear convergence. Section \ref{sec4} proposes modifications to the variable update sequence in QPADM-slack, simplifying the correction steps during Gaussian back-substitution. Numerical experiments in Section \ref{sec45} demonstrate that Gaussian back-substitution not only theoretically guarantees the linear convergence of QPADM-slack but also significantly improves computational efficiency and accuracy. Section \ref{sec46} summarizes the key findings of the study and identifies avenues for future research. The proofs of the theorems and supplementary experimental results are included in the online appendix.

\textbf{Notations}:
$\bm 0_n$ and $\bm 1_n$ represent $n$-dimensional vectors with all elements being 0 and 1, respectively.
$\bm F$ is a $(p-1) \times p$ matrix with all elements being 0, except for 1 on the diagonal and -1 on the superdiagonal.
$\bm I_n$ represents the $n$-dimensional identity matrix.
The Hadamard product is denoted by $\odot$.
The sign$(\cdot)$ function is defined component-wise such that sign$(t) = 1$ if $t > 0$, sign$(t) = 0$ if $t = 0$, and sign$(t) = -1$ if $t < 0$.
$(\cdot)_+$ signifies the element-wise operation of extracting the positive component, while $|\cdot|$ denotes the element-wise absolute value function.
For any vector $\bm u$, $\|\bm u\|_1$  and $\|\bm u\|_2$ denote the $\ell_1$ norm and the $\ell_2$ norm of $\bm u$, respectively.
$\|\bm u\|_{\bm H} := \sqrt{\bm u^\top \bm H \bm u}$ is used to denote the norm of $\bm u$ under the matrix $\bm H$, where $\bm H$ is a matrix.

\section{Preliminaries and Literature Review}\label{sec2}
 A traditional method for solving penalized quantile regression and SVMs is to transform the corresponding optimization problem into a linear program (see \cite{Zhu20031-norm}, \cite{Koenker2005Quantile}, \cite{Wu2009Variable},  \cite{Wang2012Quantile} and \cite{Zhang2016Variable}), which can then be solved using many existing optimization methods. \cite{Koenker2005AFA} proposed an interior-point method for quantile regression and penalized quantile regression. \cite{Hastie2004The}, \cite{Rosset2007Piecewise} and \cite{Li2008L1} introduced an algorithm for computing the solution path of the LASSO-penalized quantile regression and SVM, building on the LARS approach (\cite{Efron2004Least}).
However, these linear programming algorithms are known to be inefficient for high-dimensional quantile regressions and SVMs.  

Although gradient descent algorithms (\cite{Beck2009A}, \cite{Yang2015A}) and coordinate descent methods (\cite{Friedman2010Regularization}, \cite{Yang2013An}) are efficient in solving penalized smooth regression and classification problems (including some smooth SVMs), they cannot be directly applied to penalized quantile regressions and SVMs due to the non-smooth nature of quantile loss.
To extend the coordinate descent method to nonsmooth loss regressions, \cite{Peng2015An} integrated the majorization-minimization algorithm with the coordinate descent algorithm to develop an iterative coordinate descent algorithm (QICD) for solving sparse penalized quantile regression. \cite{Yi2016Semismooth} introduced a coordinate descent algorithm for penalized Huber regression and utilized it to approximate penalized quantile regression. However, these algorithms are not suitable for distributed storage and are not easy to implement in parallel.

An efficient algorithm for solving penalized quantile regressions and SVMs is the alternating direction method of multipliers (ADMM), which, owing to its split structure, is well-suited for parallel computing environments. In the following  subsection, we will review these ADMM algorithms and their parallel versions.

\subsection{ADMM}
The alternating direction method of multipliers (ADMM) is an iterative optimization method designed to solve complex convex minimization problems with linear constraints. It works by breaking down the original problem into smaller, more manageable subproblems that are easier to solve. ADMM alternates between optimizing these subproblems and updating dual variables to enforce the constraints, making it particularly useful for large-scale problems and various statistical learning applications (see \cite{Boyd2010Distributed} for more details).
Since the quantile loss function \( \rho_\tau(\cdot) \) is nonsmooth and nondifferentiable, it is necessary to introduce the linear constraint \( \bm{r} = (r_1, r_2, \dots, r_n)^\top = \tilde{\bm{y}} - \tilde{\bm{X}}\bm{\beta} \) to apply ADMM for solving problem (\ref{pqr}). To better address the subproblem involving \( \bm{\beta} \),  one can introduce the equality constraint \( \bm{z} = \bm{\beta} \). Consequently, the constrained optimization problem can be formulated as follows,
\begin{align}\label{intr1}
\min_{\bm \beta, \bm r, \bm z} & \quad  \rho_\tau (\bm r)  + {P}_\lambda(|\bm \beta|
), \notag \\
\text{s.t.} \ &  \tilde{\bm y} -  \tilde{\bm{X}} \bm z      =  \bm r ,  \ \bm{z} = \bm{\beta},
\end{align} 
where $\rho_\tau (\bm r) = \sum_{i=1}^{n} \rho (r_i)$.
The augmented Lagrangian form of (\ref{intr1}) is
\begin{align}
L_\mu(\bm \beta, \bm r, \bm z, \bm d_1, \bm e_1) &= 
\rho_\tau(\bm r)  + {P}_\lambda(|\bm \beta|)  - \bm d_1^\top (  \tilde{\bm y} - \tilde{\bm{X}} \bm z   - \bm r ) +  \frac{\mu}{2}  \|   \tilde{\bm y} - \tilde{\bm{X}} \bm z     - \bm r\|_2^2 \notag \\
& - \bm e_1^\top (\bm z - \bm \beta) +  \frac{\mu}{2} \|\bm z - \bm \beta \|_2^2,
\end{align}
where $\bm d_1$  and $\bm e_1$ are dual variables corresponding to the linear constraints, and $\mu>0$ is a given  augmented parameter.
Given \((\bm{r}^0, \bm{z}^0, \bm{d}_1^0, \bm{e}_1^0)\), the iterative scheme of ADMM for problem (\ref{intr1}) is as follows,
\begin{equation}\label{twoupadmm1}
\left\{ \begin{array}{l}
 \bm \beta^{k+1} \ \leftarrow  \mathop {\arg \min }\limits_{\bm \beta} \left\{ L_\mu(\bm \beta, \bm r^k, \bm z^k, \bm d_1^k, \bm e_1^k)  \right \};\\
\bm r^{k+1} \ \leftarrow  \mathop {\arg \min }\limits_{\bm r} \left\{L_\mu(\bm \beta^{k+1}, \bm r,  \bm z^k, \bm d_1^k, \bm e_1^k)  \right \}; \\
\bm z^{k+1} \ \leftarrow  \mathop {\arg \min }\limits_{\bm z} \left\{L_\mu(\bm \beta^{k+1}, \bm r^{k+1}, \bm z,  \bm d_1^k, \bm e_1^k)  \right \}; \\
\bm d_1^{k+1} \ \leftarrow  \bm{d}_1^{k} - \mu( \tilde{\bm y} - \tilde{\bm X} \bm z^{k+1}  - \bm r^{k+1}); \\
\bm e_1^{k+1} \ \leftarrow  \bm{e}_1^{k} - \mu(\bm z^{k+1} - \bm \beta^{k+1}).
\end{array} \right.
\end{equation}
The entire iterative process of (\ref{twoupadmm1})   is the non parallel QPADM algorithm proposed by  \cite{Yu2017A} and \cite{Wu2024Multi}.
On the other hand, \cite{Gu2018ADMM} and \cite{Wu2025Multi} introduced the linearized ADMM algorithms for solving penalized quantile regression and SVM. These methods do not introduce an auxiliary variable \(\bm{z}\), and instead linearize the quadratic function in the \(\bm{\beta}\)-subproblem to facilitate finding the solution for \(\bm{\beta}\). However, their did not extend the algorithms to handle distributed storage data, meaning there is no parallel version of their ADMM algorithms. Therefore, we will focus on discussing the parallel QPADM algorithm from \cite{Yu2017A} and \cite{Wu2025A}.

\subsection{Parallel ADMM}
When designing algorithms for distributed parallel processing, the setup generally involves a central machine and several local machines. 
Assume the data matrix $\tilde{\bm X}$ and the response vector $\tilde{\bm y}$ are distributed across $M$ local machines as follows,
\begin{align}\label{mdata}
\tilde{\bm X} =(\tilde{\bm X}_1^\top, \tilde{\bm X}_2^\top, \dots, \tilde{\bm X}_M^\top)^\top  \ \text{and} \  \ \tilde{\bm y} = (\tilde{\bm y}_1^\top, \tilde{\bm y}_2^\top, \dots, \tilde{\bm y}_M^\top)^\top,
\end{align}
where $\tilde{\bm X}_m \in \mathbb{R}^{n_m \times p}$, $\tilde{\bm y}_m \in \mathbb{R}^{n_m}$ \ and $\sum_{m=1}^M n_m = n$.
To accommodate the structure of distributed storage data, \cite{Boyd2010Distributed} introduced \(\bm \beta_m = \bm \beta, m =1,2,\dots,M \), to enable parallel processing of the data. Then, the constrained optimization problem can be formulated as 
\begin{align}\label{intr2}
\min_{\bm \beta, \bm r_m, \bm \beta_m} & \quad  \sum_{m=1}^{M}\rho_\tau (\bm r_m)  + {P}_\lambda(|\bm \beta|), \notag \\
\text{s.t.} \ \bm{\beta}_m = \bm{\beta},  \   \tilde{\bm y}_m - &   \tilde{\bm X}_m \bm \beta_m    =  \bm r_m, \  m =1,2,\dots,M,
\end{align} 
where $\rho (\bm r_m) = \sum_{i=1}^{n_m} \rho (r_i)$. The augmented Lagrangian form of (\ref{intr2}) is
\begin{align}\small
L_\mu(\bm \beta, \bm r_m, \bm \beta_m, \bm d_m, \bm e_m) & = \
 \sum_{m=1}^{M} \left[ \rho(\bm r_m)  +  \bm d_m^\top ( \tilde{\bm y}_m -  \tilde{\bm X}_m \bm \beta_m  -  \bm r_m) +  \frac{\mu}{2} \|   \tilde{\bm y}_m -  \tilde{\bm X}_m \bm \beta_m  - \bm r_m \|_2^2 \right.  \notag \\
& \left.  + \  \bm e_m^\top (\bm \beta_m- \bm \beta) +  \frac{\mu}{2}  \|\bm \beta_m - \bm \beta \|_2^2 \right]  + {P}_\lambda(|\bm \beta|),
\end{align}
where $\bm d_m$  and $\bm e_m$ are dual variables corresponding to the linear constraints. Given \((\bm{r}_m^0, \bm{\beta}_m^0, \bm{d}_m^0, \bm{e}_m^0)\), the iterative scheme of parallel ADMM for problem (\ref{intr2}) is as follows,
\begin{equation}\label{twoupadmm}
\left\{ \begin{array}{l}
\bm \beta^{k+1} \ \leftarrow  \mathop {\arg \min }\limits_{\bm \beta} \left\{ L_\mu(\bm \beta, \bm r_m^k, \bm \beta_m^k, \bm d_m^k, \bm e_m^k)  \right \};\\
\bm r_m^{k+1} \ \leftarrow  \mathop {\arg \min }\limits_{\bm r_m} \left\{L_\mu(\bm \beta^{k+1}, \bm r_m,  \bm \beta_m^k, \bm d_m^k, \bm e_m^k)  \right \}; \\
\bm \beta_m^{k+1} \ \leftarrow  \mathop {\arg \min }\limits_{\bm \beta_m} \left\{L_\mu(\bm \beta^{k+1}, \bm r^{k+1}, \bm \beta_m,  \bm d_m^k, \bm e_m^k)  \right \}; \\
\bm d_m^{k+1} \ \leftarrow  \bm{d}_m^{k} - \mu( -\tilde{\bm y}_m + \tilde{\bm X}_m \bm \beta_m^{k+1}  + \bm r_m^{k+1}); \\
\bm e_m^{k+1} \ \leftarrow  \bm{e}_m^{k} - \mu(-\bm \beta_m^{k+1} + \bm \beta^{k+1}).
\end{array} \right.
\end{equation}

This distribution enables efficient parallelization of the algorithm. Each local machine independently solves its own subproblem (e.g., $\bm r_m, \bm \beta_m, \bm u_m, \bm v_m$), while the central machine consolidates the results by updating $\bm \beta$ and coordinates updates among the local machines. This framework facilitates parallel processing and effective management of large datasets.
In subsection 3.2 of \cite{Yu2017A}, it is indicated that for convex penalties \( P_\lambda \), the QPADM converges to the solution of \eqref{intr2}. Here, we will briefly explain the proof approach. For ease of description, we will only discuss the case where \( M = 2 \); the case where \( M > 2 \) is similar.

When $M=2$, the constraint form of \eqref{intr2} is  written in matrix form as
$$\bm A_1 \bm \beta + \bm A_2 \bm r_1 + \bm A_3 \bm r_2 + \bm B_1 \bm \beta_1 +  \bm B_2 \bm \beta_2 = \bm e, $$
 where
 \[
\bm A_1 = [-\bm I_p, -\bm I_p, \bm 0^\top, \bm 0^\top]^\top, \quad
\bm A_2 = [\bm 0^\top, \bm 0^\top, \bm I_{n_1}, \bm 0^\top]^\top, \quad
\bm A_3 = [\bm 0^\top, \bm 0^\top, \bm 0^\top, \bm I_{n_2}]^\top,
\]
\[
\bm B_1 = [\bm I_p, \bm 0^\top, \bm X_1^\top, \bm 0^\top]^\top, \quad
\bm B_2 = [\bm 0^\top, \bm I_p, \bm 0^\top, \bm X_2^\top]^\top, \quad
\bm e = [\bm 0^\top, \bm 0^\top, \bm y_1^\top, \bm y_2^\top]^\top.
\] Note that \(\bm A_1 \), \(\bm A_2 \), and \(\bm A_3 \) are mutually orthogonal, as are \(\bm B_1 \) and \(\bm B_2 \). Therefore,  \( \bm \beta \), \( \bm r_1 \), and \( \bm r_2 \) can be considered as a single variable, and \(  \bm \beta_1  \) and \( \bm \beta_2 \) can also be considered as a single variable. Consequently, the QPADM  degrades to a traditional two-block ADMM algorithm in the parallel case. A similar conclusion applies to non-parallel QPADM, provided that  $\bm \beta$ and $\bm r$ are treated as a single block, and $\bm z$ is treated as a separate block. The convergence of the two-block traditional ADMM algorithm has been well-studied, thus QPADM is convergent as well. According to \cite{He2015On}, it exhibits linear convergence rate.

However, this does not apply to QPADM-slack in \cite{Guan2020An} and \cite{Fan2021Penalized} because the introduction of slack variables prevents it from being formulated as a two-block ADMM. We will provide a detailed discussion on this in subsequent sections, which also serves as the motivation for this paper.

\section{QPADM-slack and Gaussian Back Substitution}\label{sec3}
\subsection{QPADM-slack}\label{sec31}
Simulation results from  \cite{Yu2017A} indicated that QPADM may require a large number of iterations to converge for penalized quantile regressions. This poses challenges for their practical application to big data, particularly in distributed environments where communication costs are high. To address this issue, \cite{Fan2021Penalized} proposed using two sets of  nonnegative slack variables $\bm{\xi} = (\xi_1, \xi_2, \ldots, \xi_n)^\top \in \mathbb{R}^n$ and $\bm{\eta} = (\eta_1, \eta_2, \ldots, \eta_n)^\top \in \mathbb{R}^n$ to represent the quantile loss. Thus, (\ref{intr2}) can be written as
\begin{align}\label{intr3}
& \arg\min_{\bm \beta, \bm \xi_m, \bm \eta_m,\bm \beta_m}  \quad   \sum_{m=1}^{M} \left[  \tau \bm 1_{n_m}^\top \bm \xi_m + (1-\tau)  1_{n_m}^\top \bm \eta_m \right]  + {P}_\lambda(|\bm \beta |), \notag \\
& \text{s.t.} \  \bm{\beta}_m = \bm{\beta}, \ \tilde{\bm y}_m -  \tilde{\bm X}_m \bm \beta_m     =  \bm \xi_m  - \bm \eta_m,  \bm \xi_m \ge 0, \bm \eta_m \ge 0,
\end{align}
where $\bm{\xi}  = (\bm \xi_1^\top,\bm \xi_2^\top.\dots,\bm \xi_M^\top)^\top$ and $\bm{\eta}  = (\bm \eta_1^\top,\bm \eta_2^\top.\dots,\bm \eta_M^\top)^\top$. When $\tau=1$, the above equation is used by \cite{Guan2020An} to complete the classification task. 

The augmented Lagrangian form of (\ref{intr3}) is
\begin{align}\small \label{Lu}
L_\mu(\bm \beta,& \bm \xi_m, \bm \eta_m, \bm \beta_m, \bm d_m, \bm e_m) = 
 \sum_{m=1}^{M} \left[  \tau \bm 1_{n_m}^\top \bm \xi_m + (1-\tau)  \bm 1_{n_m}^\top \bm \eta_m   +  \bm d_m^\top (\bm \beta_m- \bm \beta) +  \frac{\mu}{2}  \|\bm \beta_m - \bm \beta \|_2^2  \right.  \notag \\
& \left.  +  \bm e_m^\top (  \tilde{\bm y}_m -  \tilde{\bm X}_m \bm \beta_m     - \bm \xi_m  + \bm \eta_m) +  \frac{\mu}{2} \|   \tilde{\bm y}_m -  \tilde{\bm X}_m \bm \beta_m  -   \bm \xi_m  + \bm \eta_m \|_2^2  \right]  + {P}_\lambda(|\bm \beta|). 
\end{align}

Given \(( \bm \xi_m^0, \bm \eta_m^0, \bm{\beta}_m^0, \bm{d}_m^0, \bm{e}_m^0)\), the iterative scheme of parallel ADMM for problem (\ref{intr3}) is as follows,
\begin{equation}\label{threeupadmm}
\left\{ \begin{array}{l}
\bm \beta^{k+1} \ \leftarrow  \mathop {\arg \min }\limits_{\bm \beta} \left\{ L_\mu( \bm \beta, \bm \xi_m^k, \bm \eta_m^k, \bm \beta_m^k, \bm d_m^k, \bm e_m^k)  \right \};\\
\bm \xi_m^{k+1} \ \leftarrow  \mathop {\arg \min }\limits_{\bm \xi_m \ge \bm 0} \left\{L_\mu(\bm \beta^{k+1}, \bm \xi_m,  \bm \eta_m^k, \bm \beta_m^k, \bm d_m^k, \bm e_m^k)  \right \}; \\
\bm \eta_m^{k+1} \ \leftarrow  \mathop {\arg \min }\limits_{\bm \eta_m \ge \bm 0} \left\{L_\mu(\bm \beta^{k+1}, \bm \xi_m^{k+1},  \bm \eta_m, \bm \beta_m^k, \bm d_m^k, \bm e_m^k)  \right \}; \\
\bm \beta_m^{k+1} \ \leftarrow  \mathop {\arg \min }\limits_{\bm \beta_m} \left\{L_\mu(\bm \beta^{k+1}, \bm \xi_m^{k+1},  \bm \eta_m^{k+1},  \bm \beta_m,  \bm d_m^k, \bm e_m^k)  \right \}; \\
\bm d_m^{k+1} \ \leftarrow  \bm{d}_m^{k} - \mu( - \bm \beta_m^{k+1} + \bm \beta^{k+1} );\\
\bm e_m^{k+1} \ \leftarrow  \bm{e}_m^{k} - \mu( - \tilde{\bm y}_m + \tilde{\bm X}_m \bm \beta_m + \bm \xi_m^{k+1} - \bm \eta_m^{k+1} ).
\end{array} \right.
\end{equation} 
Clearly, $\bm \beta^{k+1}$  will be updated on the central machine, while $\bm \xi_m^{k+1}, \bm \eta_m^{k+1}, \bm \beta_m^{k+1}, \bm d_m^{k+1}$ and $\bm e_m^{k+1}$ will be updated in parallel on the sub machines.
In section 3 of \cite{Fan2021Penalized}, they provided closed-form solutions for the ($\bm \beta, \bm \xi_m, \bm \eta_m, \bm \beta_m$)-subproblems, which significantly facilitated the implementation of the parallel algorithm.

The iteration process described above cannot be reduced to a two-block ADMM algorithm; it can only be reduced to a three-block ADMM algorithm. Next, we will briefly explain this point using mathematical expressions.   When $M\ge 2$, the constraint form of \eqref{intr3} is  written in matrix form as
\begin{align}\label{constr}
\bm A_1 \bm \beta + \sum_{m=2}^{M+1} \bm A_m \bm \xi_{m-1}  + \sum_{m=1}^{M}  \bm B_m \bm \eta_m  +  \sum_{m=1}^{M} \bm C_m \bm \beta_m  = \bm e, 
\end{align}
where \(\bm{A}_m\), \(\bm{B}_m\), and \(\bm{C}_m\) are block matrices partitioned into \(2M\) blocks by rows, and \(\bm{e}\) is the corresponding column vector partitioned into \(2M\) blocks.  For example,  for  \( M = 2 \),  $\bm A_1 = [- \bm I_p, - \bm I_p, \bm 0^\top,   \bm 0^\top]^\top$, $\bm A_2 = [\bm 0^\top, \bm 0^\top, \bm I_{n_1}, \bm 0^\top ]^\top$,  $\bm A_3 = [\bm 0^\top, \bm 0^\top, \bm 0^\top, \bm I_{n_2} ]^\top$,  $\bm B_1 = [\bm 0^\top, \bm 0^\top, -\bm I_{n_1}, \bm 0^\top ]^\top$, $\bm B_2 = [\bm 0^\top, \bm 0^\top, \bm 0^\top, -\bm I_{n_2} ]^\top$,  $\bm C_1 = [\bm I_p,  \bm 0^\top, \bm X_1^\top,\bm 0^\top]^\top$,  $\bm C_2= [ \bm 0^\top, \bm I_p, \bm 0^\top,   \bm X_2^\top]^\top$, and $\bm e = [\bm 0^\top, \bm 0^\top, \bm y_1^\top,  \bm y_2^\top]^\top$.  
The above six matrices cannot be divided into two sets of mutually orthogonal matrices like QPADM.
 Indeed, these six matrices can be partitioned into three mutually orthogonal groups in the following order: $\bm A = [\bm A_1, \bm A_2, \bm A_3],  \bm B = [\bm B_1, \bm B_2]$, and $\bm C = [\bm C_1, \bm C_2]$.   The same operation can also be implemented when $M>2$, that is, 
\begin{align}\label{conm} 
\bm A = [\bm A_1, \bm A_2,  \dots, \bm A_{M+1}], \  \bm B = [\bm B_1, \bm B_2,\dots,\bm B_M], \ \text{and} \  \bm C = [\bm C_1, \bm C_2,\dots,\bm C_M]. 
\end{align}
Hence, \(  \bm \beta \) and \( \bm \xi_m \)  can be considered as the first variable, \( \bm \eta_m \)  as the second variable, and \( \bm \beta_m \)  as the third variable. Thus, QPADM-slack is actually a three-block ADMM algorithm in terms of optimization form.

 \cite{Chen2016The} demonstrated that directly extending the ADMM algorithm for convex optimization with three or more separable blocks may not guarantee convergence, and they provided an example of divergence.  As a result, QPADM-slack  does not have a theoretical guarantee of convergence even in the convex case.

 \subsection{Gaussian Back Substitution}\label{sec32}
\cite{He2012Alternating}  pointed out that although the ADMM algorithm with three blocks cannot be theoretically proven to converge, the algorithm can be corrected for some iterative solutions through simple Gaussian back substitution, thereby making the algorithm convergent.  Because the correction matrix is an upper triangular block matrix, it is called \textit{Gaussian back substitution}. The technique of Gaussian back substitution was extensively applied to various statistical and machine learning domains, including but not limited to the works of \cite{Ng2013Coupled}, \cite{He2013Linearized}, \cite{He2017Convergence}, and \cite{He2018A}.

Let 
\begin{align}\label{abcd}
 \bm a = (\bm \beta, \bm \xi_1,  \dots, \bm \xi_M ),  \ \bm b =( \bm \eta_1,  \dots, \bm \eta_M) \ \text{and} \ \bm c = ( \bm \beta_1,  \dots, \bm \beta_M),
\end{align}
and $\bm d = (\bm d_1,\dots,\bm d_M,\bm e_1,\dots,\bm e_M)$.
To facilitate the description of the correction process (Gaussian back substitution), we define the sequence of \( k+1 \) iterations \(({\bm b}^{k+1}, {\bm c}^{k+1})\), generated by (\ref{threeupadmm}), as \((\tilde{\bm b}^k, \tilde{\bm c}^k)\).  Then, the  Gaussian back substitution of QPADM-slack is defined as 
\begin{align}\label{gbs}
 \begin{bmatrix}
\bm b^{k+1}  \\
\bm c^{k+1} 
\end{bmatrix} =  \begin{bmatrix}
 {\bm b}^k  \\
{\bm c}^k 
\end{bmatrix} -    \begin{bmatrix}
v  \bm I & -v (\bm B^\top \bm B)^{-1} \bm B^\top \bm C \\
\bm 0 & v \bm I
\end{bmatrix}  \begin{bmatrix}
{\bm b}^k  -   \tilde{\bm b}^k  \\
{\bm c}^k  - \tilde{\bm c}^k
\end{bmatrix},
\end{align}
where $\nu \in (0,1)$.
Note that $\bm B^\top \bm B$  is an identity matrix, and thus $$(\bm B^\top \bm B)^{-1} \bm B^\top \bm C = \bm B^\top \bm C  = \begin{bmatrix}
  -  \bm X_1 & 0  & \cdots & 0 \\
    0 & -\bm X_2  & \cdots & 0 \\
    0 & 0  & \cdots & 0 \\
    \vdots & \vdots  & \ddots & \vdots \\
    0 & 0 & \cdots &- \bm X_M
\end{bmatrix}, $$
it follows that
\begin{equation}\label{cr}
\begin{bmatrix}
\bm{b}^{k+1} \\
\bm{c}^{k+1}
\end{bmatrix}
=
\begin{bmatrix}
\bm \eta_m^{k+1} \\
\bm \beta_m^{k+1}
\end{bmatrix}
=
\begin{bmatrix}
(1-\nu){\bm \eta}_m^k + \nu  \tilde{\bm \eta}_m^k - \nu \bm X_m (\bm \beta_m^k - \tilde{\bm \beta}_m^k) \\
(1-\nu){\bm \beta}_m^k + \nu  \tilde{\bm \beta}_m^k.
\end{bmatrix}.
\end{equation}

We summarize the process of QPADM-slack with Gaussian back substitution (QPADM-slack(GB)) in Algorithm \ref{alg1}. Its parallel algorithm diagram is shown in Figure \ref{fig41}. Algorithm  \ref{alg1} differs from QPADM-slack only in the correction steps performed in each local machine, specifically steps 3 and 4 in \textbf{Local machines}. This correction process involves only linear operations, with the heaviest computational load being the matrix-vector multiplication in the \(\bm \eta_m^{k+1}\) correction step.  Next, we will establish the global convergence of QPADM-slack(GB) and its linear convergence rate. The proof of the following theorem is provided in Appendix \ref{4B}.

\begin{figure}
    \centering
    \includegraphics[width=1\linewidth]{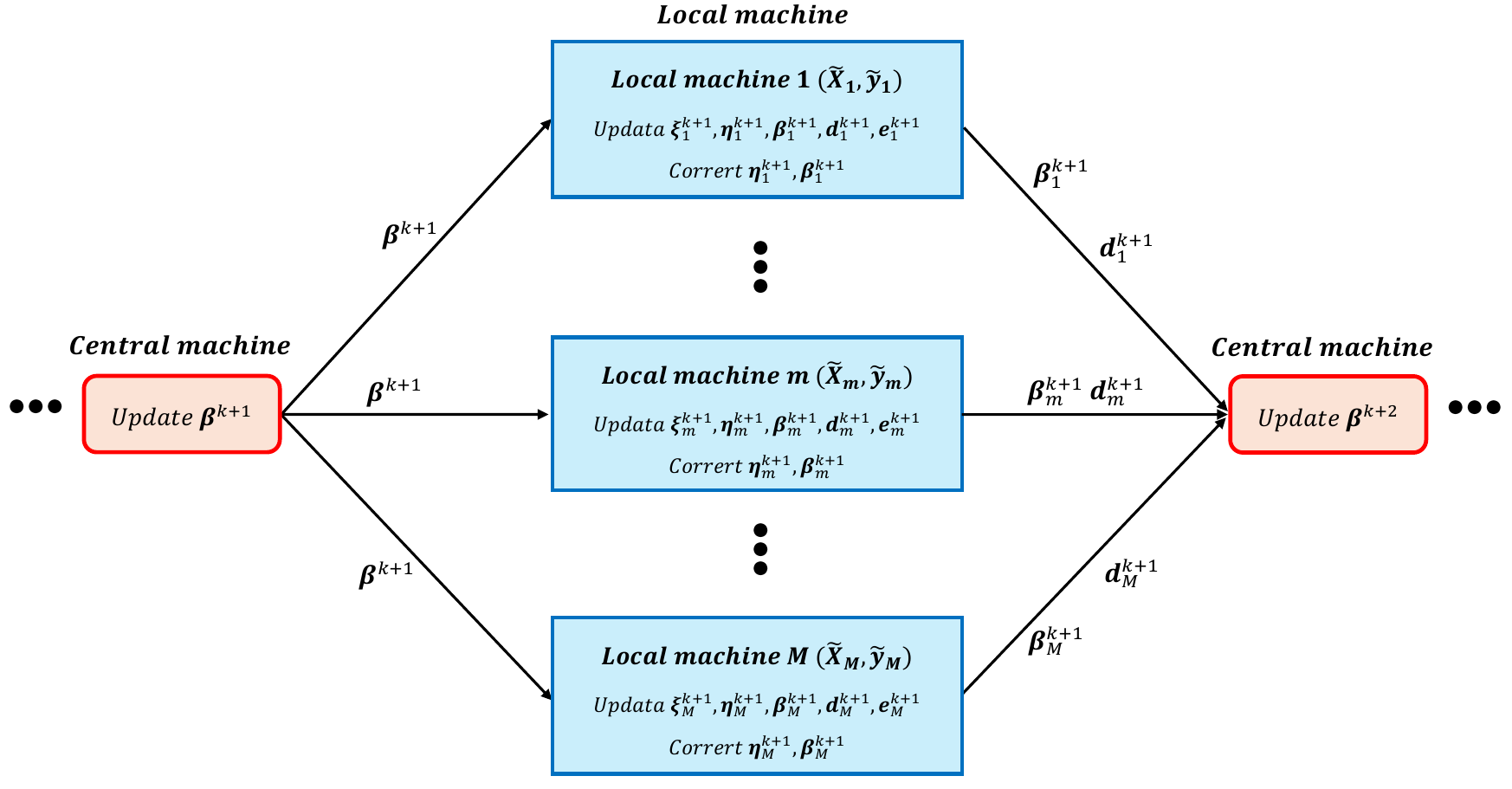}
    \caption{Schematic diagram of QPADM-slack(GB) algorithm.}
    \label{fig41}
\end{figure}

\begin{thm}\label{TH1}
Let the sequence \( \left\{ \bm g^k= ( \bm{b}^k, \bm c^k, \bm d^k) \right\} \) be generated by QPADM-slack(GB).
\vspace{-1.5em}
\begin{enumerate}
\item (Algorithm global convergence). It converges to some  \( \bm g^\infty =  ( \bm{b}^\infty, \bm c^\infty, \bm d^\infty ) \)  that belongs to \( \mathcal{V}^* \), where \( \mathcal{V}^* = \{ (\bm{b}^*, \bm{c}^*, \bm{d}^*) \mid (\bm{a}^*, \bm{b}^*, \bm{c}^*, \bm{d}^*) \in \Omega^* \} \),  and $ \Omega^*$ is the set of optimal solutions for (\ref{intr3}).
\item (Linear convergence rate).  For any integer \(K > 0\), we have
\begin{align}\label{p25}
\|\bm g^K - \bm g^{K+1} \|_{\bm{H}}^2 \leq \frac{1}{ c_0\left(K+1\right)}  \|\bm g^0 - \bm g^*\|_{\bm{H}}^2,
\end{align}
where $c_0$ is a positive constant and $\bm H$ is a  positive definite matrix.
\end{enumerate}
\end{thm}
\begin{rem}
 It is clear that \(\|\bm{g}^0 - \bm{g}^*\|_{\bm{H}}^2\) is a positive constant. Therefore, \(\|\bm{g}^K - \bm{g}^{K+1}\|_{\bm{H}}^2 \leq \mathcal{O}\left(\frac{1}{K}\right)\), which is known as a linear convergence rate. In addition, during the proof of this theorem, we demonstrated that \( \|\bm{g}^k - \bm{g}^{*}\|_{\bm{H}}^2 \) and \(  \|\bm{g}^k - \bm{g}^{k+1}\|_{\bm{H}}^2  \) are monotonically nonincreasing, that is, \( \|\bm{g}^{k+1} - \bm{g}^{*}\|_{\bm{H}}^2   \leq \|\bm{g}^k - \bm{g}^{*}\|_{\bm{H}}^2 \) and \( \|\bm{g}^{k+1} - \bm{g}^{k+2}\|_{\bm{H}}^2   \leq \|\bm{g}^k - \bm{g}^{k+1}\|_{\bm{H}}^2 \),  see the propositions in the Appendix \ref{4B}.
\end{rem}

\begin{algorithm}\small
\caption{\small{QPADM-slack with Gaussian back substitution}}
\label{alg1}
\begin{algorithmic}
\STATE {\textbf{Input:} $\bullet$ Central machine: $\mu, M, \bm \lambda$.\\
\qquad \  \ \ \ \ $\bullet$ The $m$-th local machine:  $\boldsymbol{X}_m,\boldsymbol{y}_m$; $\mu$, $\nu$, $\tau$; $\bm \beta_m^0$, $\bm \eta_m^0$, $\bm d_m^0$, $\bm e_m^0$.}
\STATE {\textbf{Output:} the total number of iterations $K$,  $\boldsymbol{\beta}^K$. }
\STATE {\textbf{while} not converged \textbf{do}}
\STATE {\ \textbf{Central machine}: 1. Receive ${\bm \beta}_m^{k}$ and $\bm d_m^k$ transmitted by $M$ local machines,\\ \qquad \qquad \qquad  \qquad \quad \ 2. Update $\bm \beta^{k+1}$ as QPADM-slack, \\
\qquad \qquad \qquad  \qquad \quad \ 3. Send $\bm \beta^{k+1}$ to the local machines.
}
\STATE {\ \textbf{Local machines}: \ \  for $m =1 ,2, \dots, M$ (in parallel) \\
\qquad \qquad \qquad  \qquad \quad \ 1. Receive $\bm \beta^{k+1}$  transmitted by the central machine, \\
\qquad \qquad \qquad  \qquad \quad \ 2. Update $\bm \xi_m^{k+1}$, ${\bm \eta}_m^{k+1}$,  $ \bm {\beta}_m^{k+1}$,  $\bm d_m^{k+1}$ and $\bm e_m^{k+1}$ as QPADM-slack,\\
\qquad \qquad \qquad  \qquad \quad \ 3. Let $\tilde{\bm \eta}_m^{k}  = \bm \eta_m^{k+1}$ and $\tilde{\bm \beta}_m^{k}  = \bm \beta_m^{k+1}$,\\
\qquad \qquad \qquad  \qquad \quad \ 4. Correct $\bm \eta_m^{k+1} $ and $\bm \beta_m^{k+1}$ according to \textbf{Gaussian back substitution} in (\ref{cr}), \\
\qquad \qquad \qquad  \qquad \quad \ 5. Send ${\bm \beta}_m^{k+1}$ and $\bm d^{k+1}_m$  to the central machine.
}
\STATE {\textbf{end while}}
\STATE {\textbf{return} solution}.
\end{algorithmic}
\end{algorithm}

\subsection{Nonconvex Extension and Local Linear Approximation}
In recent years, nonconvex penalized quantile regression has been extensively studied, both theoretically and algorithmically, see \cite{Wang2012Quantile}, \cite{Fan2014Strong}, \cite{Gu2018ADMM}, \cite{Fan2021Penalized}, \cite{Wang2024Analysis}. However, there has been little research on quantile loss SVMs with nonconvex regularization terms. In this subsection, we propose the nonconvex quantile loss SVMs, defined as 
\begin{equation}\label{pinsvm2}
\begin{array}{l}
\mathop {\arg \min }\limits_{\boldsymbol{\beta},{\xi} \ge \boldsymbol 0}  \sum\limits_{i = 1}^n {{{\xi} _i}} +  P_{\bm \lambda}(|\bm \beta|)    \\
\text{s.t.}\;\; \sum\limits_{j = 1}^p y_ix_{ij}\beta_j  \ge 1 - {1 \over \tau}{{\xi} _i},\quad i = 1,2,\cdots n,\\
\ \ \ \ \ \  \sum\limits_{j = 1}^p y_ix_{ij}\beta_j   \le 1 + {1 \over (1-\tau)}{{\xi} _i}, \ i = 1,2,\cdots n.
\end{array}
\end{equation}
When \( \tau = 1 \), the second inequality is always satisfied, and thus the aforementioned SVM simplifies to the nonconvex hinge SVM addressed in \cite{Zhang2016Variable} and \cite{Guan2018An}. It is not difficult to see that \eqref{pinsvm2} can be included in \eqref{pqr}. In this paper, we mainly consider two popular nonconvex regularizations, SCAD penalty and MCP penalty. 

In high-dimensional penalized linear regression and classification models, convex regularization terms like adaptive LASSO and elastic net ensure global optimality and computational efficiency, while non-convex regularization terms may provide better estimation and prediction performance but pose computational challenges due to the lack of global optimality. 
Fortunately, for quantile regression models and SVM models with nonconvex penalties, \cite{Fan2014Strong} and \cite{Zhang2016Variable} respectively pointed out that these problems can be uniformly solved by combining the local linear approximation (LLA, \cite{Zou2008One}) method with an effective solution for the weighted \(\ell_1\) penalized quantile regression and SVM.

LLA involves using the first-order Taylor expansion of the nonconvex regularization term to replace the original nonconvex regularization term, that is
\begin{align}\label{one}
P_{\bm \lambda}(|\bm \beta|) \approx  P_{\bm \lambda}(|\bm \beta^l|) + \nabla P_{\bm \lambda}(|\bm \beta^l|)^\top (|\bm \beta| - |\bm \beta^l|), \ \text{for} \ \bm \beta  \approx \bm \beta^l,
\end{align}
where  $\bm \beta^l$ is the solution from the last iteration, and $\nabla P_{\bm\lambda}(|\bm \beta^l|) = (\nabla P_{\bm \lambda_1}(|\bm \beta_1^l|), \nabla P_{\bm \lambda_2}(|\bm \beta_2^l|), \dots, \nabla P_{\bm \lambda_p}(|\bm \beta_p^l|)) ^\top$.

$\bullet$ For SCAD, we have \begin{align}\label{scad}
\nabla P_{\bm \lambda_j}(|\bm \beta_j|) = \begin{cases}
\bm \lambda_j, & \text{{if }} |\bm \beta_j| \leq  \bm \lambda_j, \\
\frac{a \bm\lambda_j -  |\bm \beta_j|}{a-1}, & \text{{if }} \bm \lambda_j < |\bm \beta_j|  < a  \bm \lambda_j, \\
0 , & \text{{if }} |\bm \beta_j| \ge a  \bm \lambda_j. \\
\end{cases}
\end{align}

$\bullet$ For MCP, we have \begin{align}\label{mcp}
\nabla P_{\bm \lambda_j}(|\bm \beta_j|)  =
\begin{cases}
\bm\lambda_j  - \frac{|\bm \beta_j|}{a}, & \text{{if }}  |\bm \beta_j|  \le a  \bm \lambda_j,\\
0, & \text{{if }}  |\bm \beta_j|  > a \bm \lambda_j.\\
\end{cases}
\end{align}
Then, the nonconvex penalized quantile regressions and SVMs can be written as
\begin{align}\label{41}
 \arg \min_{\bm \beta} \left\{ \sum_{i=1}^{n} \rho_{\tau} \left( \tilde{y_i} - \tilde{\boldsymbol{x}}_i^\top \bm \beta  \right) + P_{\bm\lambda} (|\bm \beta|) \right\}.
\end{align}
 By substituting equation (\ref{one}) into equation (\ref{41}), we can obtain the following optimized form in a weighted manner,
\begin{align}\label{we}
\boldsymbol{\beta}^{l+1} = \mathop {\arg \min }\limits_{\boldsymbol{\beta}} \left\{ \sum_{i=1}^{n} \rho_{\tau} \left( \tilde{y_i} - \tilde{\boldsymbol{x}}_i^\top \bm \beta  \right)  + \sum_{j=1}^{p}\nabla P_{\bm\lambda_j}(|\beta_j^{l}|)|\beta_j|   \right \}.
\end{align}
Note that we only need to make a small change to solve this weighted  optimization form using Algorithm \ref{alg1}. This change only requires replacing $\bm \lambda$ with $\nabla P_{\bm \lambda}(|\bm \beta^{l}|)$. 

To address nonconvex penalized quantile regressions and SVMs through the LLA algorithm, it is crucial to identify a suitable initial value. Following the guidance of \cite{Zhang2016Variable} and \cite{Gu2018ADMM}, we can utilize the solution derived from \( P_{\bm\lambda}(| \bm \beta|) = \| \bm\lambda \odot \bm \beta \|_1 \) in (\ref{41}) as the initial starting point. Subsequently, the solution to (\ref{41}) is obtained by iteratively solving a series of weighted \(\ell_1\) penalized quantile regressions and SVMs. The comprehensive iterative steps of this approach are systematically outlined in Algorithm \ref{alg2}. 
\begin{algorithm}\small
\caption{\small{QPADM-slack(GB) for nonconvex penalized quantile  regressions and SVMs}}
\label{alg2}
\begin{algorithmic}
\STATE {1. Initialize $\bm \beta$ with  $\bm \beta^1$, where $\bm \beta^1$ is obtained by Algorithm \ref{alg1}.}
\STATE {2. For $l=1,2,\dots,L$, continue iterating the LLA iteration until convergence is achieved. }
\STATE {\quad \ 2.1. Compute the weights  $\nabla P_{\bm \lambda}(|\bm \beta^{l}|)=(\nabla P_{\bm \lambda_1}(|\bm \beta_1^{l}|), \nabla P_{\bm \lambda_2}(|\bm \beta_2^{l}|), \dots, \nabla P_{\bm \lambda_p}(|\bm \beta_p^{l}|)) ^\top$ by (\ref{scad}) or  (\ref{mcp}),}
\STATE { \quad  \ 2.2. Solve the weighted problem in (\ref{we}) by modified Algorithm \ref{alg1} with $\bm \lambda$  replacing with $\nabla P_{\bm \lambda}(|\bm \beta^{l}|)$. Let this solution be denoted as $\bm \beta^{l+1}$.}
\end{algorithmic}
\end{algorithm}

The preceding discussion reveals that resolving nonconvex penalized quantile regressions and SVMs requires multiple iterations of weighted \(\ell_1\) penalized models. Theoretically, as shown by \cite{Fan2014Strong} and \cite{Zhang2016Variable}, only two to three iterations suffice to find the Oracle solution for (\ref{41}). In implementing Algorithm \ref{alg2}, we utilize the warm-start technique (\cite{Friedman2010Regularization}), initializing \(\bm \beta^{l+1}\) with \(\bm \beta^{l}\), significantly reducing step 2.2's iteration count, often leading to convergence within just a few to a dozen steps.

\section{Modified QPADM-slack and Gaussian Back Substitution}\label{sec4}
The iteration order of variables in QPADM-slack   is 
\begin{align}\label{order}
\cdots \to {\bm \beta \to  \bm \xi_m \to \bm \eta_m \to \bm \beta_m} \to \cdots,
\end{align}
 and  its Gaussian back substitution is  $\bm \eta_m^{k+1} =  (1-v){\bm \eta}_m^k + v  \tilde{\bm \eta}_m^k + \bm X_m (\bm \beta_m^k - \tilde{\bm \beta}_m^k)$  \text{and}  $\bm \beta_m^{k+1} =  (1-v){\bm \beta}_m^k + v  \tilde{\bm \beta}_m^k$. We should note that the correction step for  $\bm \eta_m$ involves matrix-vector multiplication, which can impose an additional computational burden when both \( n_m \) and \( p \) are large.
\subsection{Modified QPADM-slack}\label{sec41}
To address this issue, we suggest changing the variable update order of  QPADM-slack in (\ref{order}) to
\begin{align} 
\cdots \to \bm \beta_m \to  \bm \xi_m \to \bm \eta_m \to \bm \beta \to \cdots. 
\end{align}
We will explain the reason for adopting this order below.
From the discussion in Section \ref{sec32}, it is evident that the first block among the three blocks of primal variables to be updated does not require correction. To avoid matrix operations in the correction step,  \(\bm{\beta_m}\) should be placed in the first block because only the constraint matrix corresponding to \(\bm{\beta_m}\) is not composed of identity and zero matrices. During the Gaussian back substitution process, we need to compute the inverse of \(\bm{B}^\top \bm{B}\). To keep this inverse as simple as possible, we must place either \(\xi_m\) or \(\eta_m\) in the second block. Here, we choose to place \(\xi_m\) in the second block in sequence. The remaining \(\eta_m\), \( \bm{\beta} \) form the third block.

Thus, with \((\bm{\xi}_m^0, \bm \eta_m^0, \bm{\beta}^0, \bm{d}_m^0, \bm{e}_m^0)\), the iterative scheme of modified QPADM-slack (M-QPADM-slack) for problem (\ref{intr3}) is as follows,
\begin{equation}\label{mqpadm}
\left\{
\begin{aligned}
\bm \beta_m^{k+1} &\leftarrow \mathop {\arg \min }\limits_{\bm \beta_m} \left\{L_\mu(\bm \beta_m, \bm \xi_m^{k},  \bm \eta_m^{k},   \bm \beta^k,  \bm d_m^k, \bm e_m^k)  \right\}; \\
\tilde{\bm \xi}_m^{k} &\leftarrow \mathop {\arg \min }\limits_{\bm \xi_m \ge \bm 0} \left\{L_\mu(\bm \beta_m^{k+1}, \bm \xi_m,  \bm \eta_m^k, \bm \beta^k, \bm d_m^k, \bm e_m^k)  \right\} ;\\
\tilde{\bm \eta}_m^{k} &\leftarrow \mathop {\arg \min }\limits_{\bm \eta_m \ge \bm 0} \left\{L_\mu(\bm \beta_m^{k+1}, \tilde{\bm \xi}_m^{k},  \bm \eta_m,  \bm \beta^k, \bm d_m^k, \bm e_m^k)  \right\};  \\
\tilde{\bm \beta}^{k} &\leftarrow \mathop {\arg \min }\limits_{\bm \beta} \left\{ L_\mu(\bm \beta_m^{k+1}, \tilde{\bm \xi}_m^{k}, \tilde{\bm \eta}_m^{k},  \bm \beta, \bm d_m^k, \bm e_m^k)  \right\}; \\
\bm d_m^{k+1} &\leftarrow \bm{d}_m^{k} - \mu \left( - \bm \beta_m^{k+1} + \tilde{\bm \beta}^{k} \right);  \\
\bm e_m^{k+1} &\leftarrow \bm{e}_m^{k} - \mu \left( - \tilde{\bm y}_m + \tilde{\bm X}_m \bm \beta_m^{k+1}  + \tilde{\bm \xi}_m^{k} - \tilde{\bm \eta}_m^{k} \right);
\end{aligned}
\right.
\end{equation}
where $ L_\mu$ is  defined as in  (\ref{Lu}).  $\bm \beta_m^{k+1}$, $\tilde{\bm \xi}_m^{k}$, $\tilde{\bm \eta}_m^{k}$, $\bm d_m^{k+1}$ and $\bm e_m^{k+1}$ will be updated in parallel by each sub machine, while $\tilde{\bm \beta}^{k}$ will be updated on  the central machine.  However, this update sequence will cause a new problem, that is, the updates of local sub machines become incoherent, necessitating that the central machine first completes the update of $\tilde{\bm \beta}^{k}$ before proceeding with the update of $\bm d_m^{k+1}$. This will result in an additional round of communication between the central machine and the local machines, thereby reducing the efficiency of the algorithm. This issue will not occur in \eqref{threeupadmm} because the updates of variables on its local sub machines are coherent. Fortunately, since the updates of each sub problem in (\ref{mqpadm}) do not strictly depend on the previously updated variables, this issue can be resolved by adjusting the update positions of the dual variables $\bm d_m$ and $\bm e_m$.

For the convenience of describing the correction steps,  we  need to define 
\begin{align}\label{mabcd}
 \bm a =  ( \bm \beta_1,  \dots, \bm \beta_M), \   \bm b = (\bm \xi_1,  \dots, \bm \xi_M ) \ \text{and} \ \bm c =  ( \bm \eta_1,  \dots, \bm \eta_M, \bm \beta).
\end{align}
Correspondingly,  \[
\bm A = \begin{bmatrix}
\bm I_{p} &   & \\
 &  \ddots & \\
 & & \bm I_p\\
\bm X_{1} &   & \\
 &  \ddots & \\
 & & \bm X_M 
\end{bmatrix}, \bm B = \begin{bmatrix}
\bm 0 &   & \\
 &  \ddots & \\
 & & \bm 0\\
- \bm I_{n_1} &   & \\
 &  \ddots & \\
 & & -\bm I_{n_M} 
\end{bmatrix}  \text{and}  \ \bm C = \begin{bmatrix}
&   & & -\bm I_p \\
 &   & & \vdots\\
 & &  &-\bm I_p \\
\bm I_{n_1} &   & & \\
 &  \ddots & & \\
 & & \bm I_{n_M} & 
 \end{bmatrix}.
\]
 It is obvious that the matrices $\bm A$ and $\bm B$ are $2M \times M$ partitioned, while  $\bm C$ is    $2M \times (M+1)$ partitioned.  \( \bm A_m \),  \( \bm B_m \) and \( \bm C_m \) correspond to the block matrices of the \( m \)-th columns of \( \bm A \),  \( \bm B \) and \( \bm C \), respectively. 
We also define the sequence of \( k+1 \) iterations \(({\bm b}^{k+1}, {\bm c}^{k+1})\), generated by (\ref{mqpadm}), as \((\tilde{\bm b}^k, \tilde{\bm c}^k)\).  According to  (\ref{gbs}), it follows from   \[
(\bm B^\top \bm B)^{-1} \bm B^\top \bm C =   \begin{bmatrix}
-\bm I_{n_1} &   & & & \bm 0 \\
 &- \bm I_{n_2} &  &  & \bm 0\\
 &  & \ddots &  &  \vdots \\
 & &  &- \bm I_{n_M}  & \bm 0
 \end{bmatrix}
\] that

\begin{equation}\label{cr22}
\begin{bmatrix}
\bm{b}^{k+1} \\
\bm{c}^{k+1}
\end{bmatrix}
=
\begin{bmatrix}
\bm{\xi}_m^{k+1} \\
\bm{\eta}_m^{k+1} \\
\bm{\beta}^{k+1}
\end{bmatrix}
=
\begin{bmatrix}
(1-\nu) \bm{\xi}_m^k + \nu \tilde{\bm{\xi}}_m^k - \nu (\bm{\eta}_m^k - \tilde{\bm{\eta}}_m^k) \\
(1-\nu) \bm{\eta}_m^k + \nu \tilde{\bm{\eta}}_m^k \\
(1-\nu) \bm{\beta}^k + \nu \tilde{\bm{\beta}}^k
\end{bmatrix},
\end{equation}
where $\nu \in (0,1)$. Unlike the correction step in (\ref{cr}), the correction step in (\ref{cr22}) involves only simple addition and subtraction operations, without matrix-vector multiplications. As a result, the computational burden of this correction is less affected by large values of \(n_m\) and \(p\).

\begin{figure}[H]
    \centering
    \includegraphics[width=1\linewidth]{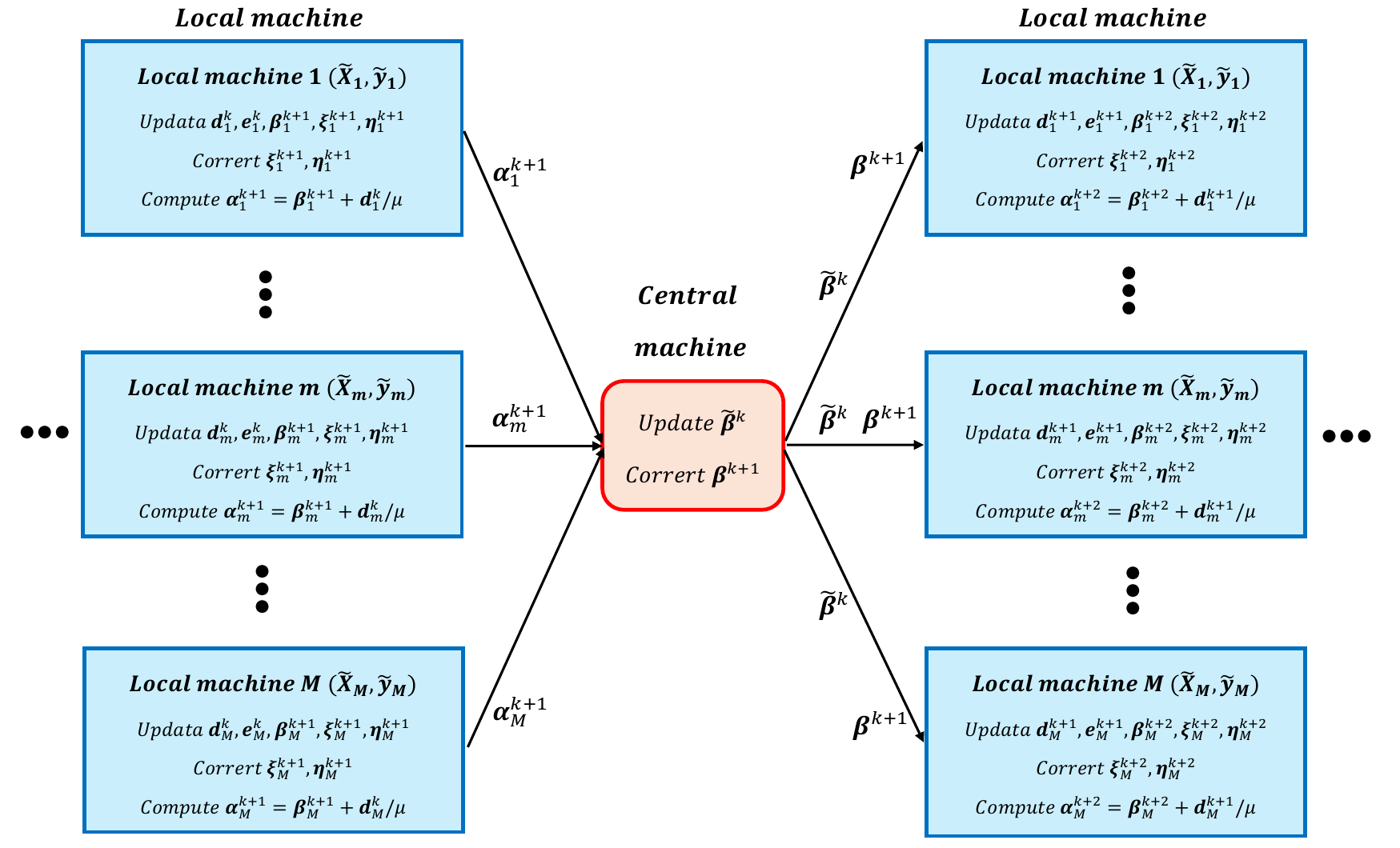}
    \caption{Schematic diagram of M-QPADM-slack(GB) algorithm.}
    \label{fig42}
\end{figure}

We summarize the process of M-QPADM-slack with Gaussian back substitution (M-QPADM-slack(GB)) in Algorithm \ref{alg3}, and its parallel algorithm diagram is shown in Figure \ref{fig42}.  In Algorithm \ref{alg3}, prior to commencing the $(k+1)$-th iteration, it is essential to first update the dual variables ($\bm{d}_m^k$ and $\bm{e}_m^k$) from the $k$-th iteration. This is achievable due to the sequence outlined in (\ref{mqpadm}), where the dual variable updates are positioned at the conclusion of each iteration. Specifically, the updates for the primal variables in the $(k+1)$-th iteration depend only on the dual variables $\bm{d}_m^k$ and $\bm{e}_m^k$ from the previous iteration, and not on the newly updated $\bm{d}_m^{k+1}$ and $\bm{e}_m^{k+1}$. By relocating the updates of $\bm{d}_m^k$ and $\bm{e}_m^k$, which are initially set to occur at the end of the $k$-th iteration, to the start of the $(k+1)$-th iteration, the incoherence in the sequence of iteration variables is effectively addressed. It is crucial to note that this rearrangement does not modify the iterative process delineated in (\ref{mqpadm}) within the execution of Algorithm \ref{alg3}.

For some nonconvex regularization terms, M-QPADM-slack(GB) can also be embedded in Algorithm \ref{alg2} to solve, simply replacing Algorithm \ref{alg1} in its step 2.2 with Algorithm \ref{alg3}. In addition,  Algorithm \ref{alg3} also inherits the convergence result of Theorem \ref{TH1}, with \(\bm b\) and \(\bm c\) in (\ref{abcd}) replaced by those in (\ref{mabcd}).

\begin{algorithm}\small
\caption{\small{Modified QPADM-slack with Gaussian Back Substitution}}
\label{alg3}
\begin{algorithmic}
\STATE {\textbf{Input:} $\bullet$ Central machine: $\mu, \nu,  M, \bm \lambda$; $\tilde{\bm \beta}^{-1}$ and ${\bm \beta}^{0}$.\\
\qquad \  \ \ \ \ $\bullet$ The $m$-th local machine:  $\boldsymbol{X}_m,\boldsymbol{y}_m$; $\mu$, $\nu$, $\tau$; $\bm \beta_m^0$, $\bm \xi_m^{-1}$, $\bm \xi_m^{0}$,$\bm \eta_m^{-1}$, $\bm \eta_m^{0}$, $\bm d_m^{-1}$, $\bm e_m^{-1}$.}
\STATE {\textbf{Output:} the total number of iterations $K$,  $\boldsymbol{\beta}^K$. }
\STATE {\textbf{while} not converged \textbf{do}}
\STATE {\ \textbf{Local machines}: \ \  for $m =1 ,2, \dots, M$ (in parallel) \\
\qquad \qquad \qquad  \qquad \quad \ 1. Receive $\tilde{\bm \beta}^{k-1}$ and ${\bm \beta}^{k}$  transmitted by the central machine, \\
\qquad \qquad \qquad  \qquad \quad \ 2. Update  $\bm d_m^{k} = \bm d_m^{k-1} - \mu (-\bm \beta_m^{k} + \tilde{\bm \beta}^{k-1} ) $ and $\bm e_m^{k} = \bm e_m^{k-1} - \mu( -\tilde{\bm y}_m + \tilde{\bm X}_m \bm \beta_m^{k} + \tilde{\bm \xi}_m^{k-1}  - \tilde{\bm \eta}_m^{k-1}  ) $,    \\
\qquad \qquad \qquad  \qquad \quad \ 3. Update  $\bm \beta_m^{k+1}$,  $\tilde{\bm \xi}_m^{k}$ and $\tilde{\bm \eta}_m^{k}$  according to  \eqref{betam}-\eqref{eta},\\
\qquad \qquad \qquad  \qquad \quad \ 4. Correct $\bm \xi_m^{k+1} $ and $\bm \eta_m^{k+1}$ according to (\ref{cr22}), \\
\qquad \qquad \qquad  \qquad \quad \ 5. Compute $\bm \alpha_m^{k+1} ={\bm \beta}_m^{k+1} + \bm d_m^{k}/\mu $ ,\\
\qquad \qquad \qquad  \qquad \quad \ 6. Send $\bm \alpha_m^{k+1}$ to the central machine.
}
\STATE {\ \textbf{Central machine}: 1. Receive $\bm \alpha_m^{k+1}$ transmitted by $M$ local machines,\\ \qquad \qquad \qquad  \qquad \quad \ 2. Update $\tilde{\bm \beta}^{k}$  according to \eqref{stbeta}, \\
\qquad \qquad \qquad  \qquad \quad \ 3.  Correct $\bm \beta^{k+1} $ according to  (\ref{cr22}),\\
\qquad \qquad \qquad  \qquad \quad \ 4. Send $\tilde{\bm \beta}^{k}$  and  ${\bm \beta}^{k+1}$ to the local machines.
}
\STATE {\textbf{end while}}
\STATE {\textbf{return} solution}.
\end{algorithmic}
\end{algorithm}

\subsection{Details of Iteration}
In this subsection, we describe the details of solving each subproblem in (\ref{mqpadm}). In fact, our M-QPADM-slack (\ref{mqpadm}) differs from the QPADM-slack in (\ref{threeupadmm}) only in the order of the iteration variables. Therefore, the solutions for each subproblem in (\ref{mqpadm}) are generally the same as those used in Section \ref{sec3}, except for the order of the iteration variables. For completeness, we also provide detailed iteration steps.

For the update of the subproblems $\bm \beta^{k+1}_m$,  $\tilde{\bm \xi}_m^{k}$ and $\tilde{\bm \eta}_m^{k}$  in  (\ref{mqpadm}), the minimization problem is quadratic and differentiable, allowing us to solve the subproblem by solving
\begin{align}
\bm \beta_m^{k+1} &= (\tilde{\bm X}_m^\top \tilde{\bm X}_m + \bm I_p)^{-1} \left[ (\bm \beta^{k} - {\bm d_m^k}/{\mu}) + \tilde{\bm X}_m^\top (\tilde{\bm y}_m - \bm \xi_m^{k} + \bm \eta_m^{k} + {\bm e_m^k}/{\mu}) \right], \label{betam}\\
\tilde{\bm \xi}_m^{k} &= \max \left\{ \tilde{\bm y}_m - \tilde{\bm X}_m \bm \beta_m^{k+1} + \bm \eta_m^{k} + {\bm e_m^k}/{\mu} - {\tau \bm 1_{n_m}}/{\mu}, \bm 0 \right\}, \label{xi}\\
\tilde{\bm \eta}_m^{k} &= \max \left\{ {(\tau-1) \bm 1_{n_m}}/{\mu} - (\tilde{\bm y}_m - \tilde{\bm X}_m \bm \beta_m^{k+1} - \tilde{\bm \xi}_m^{k} + {\bm e_m^k}/{\mu}), \bm 0 \right\}. \label{eta}
\end{align}
For the first equation,  \cite{Yu2017A}  suggested using the Woodbury matrix identity $(\tilde{\bm X}_m^\top \tilde{\bm X}_m + \bm I_p)^{-1} =  \bm I_p - \tilde{\bm X}_m^\top ( \bm I_{n_m} + \tilde{\bm X}_m \tilde{\bm X}_m^\top)^{-1} \tilde{\bm X}_m$. This method is actually very practical when the size of $n_m$ is small because the inverse only needs to be computed once throughout the ADMM iteration. 

For the subproblem of updating   $ \tilde{\bm \beta}^{k}$   in  (\ref{mqpadm}),  by rearranging the optimization equation and omitting some constant terms,   we get
\begin{align}\label{stbeta}
\tilde{\bm \beta}^{k} = \arg \min_{\boldsymbol\beta}  \left\{  P_{\bm \lambda}(|\bm \beta|)  + \frac{\mu M}{2} \left\| \bm \beta - \sum_{m=1}^M \bm \alpha_m^{k+1}/M  \right\|_2^2 \right\},
\end{align}
where  $\bm \alpha_m^{k+1} ={\bm \beta}_m^{k+1} + \bm d_m^{k}/\mu $. It is clear that the above expression represents a proximal operator, applicable to most convex regularization terms, such as LASSO and group LASSO (see \cite{Boyd2010Distributed}), elastic-net (see \cite{Parikh2013Proximal}), and sparse group LASSO (see \cite{Sprechmann2010CHiLasso}). Since the entire objective function is additive, the optimization problem can be transformed into smaller, independent subproblems during implementation.

\section{Numerical Results}\label{sec45}
This section is dedicated to demonstrating the efficacy of the proposed algorithm in solving classification and regression problems. It focuses on evaluating the algorithm's model selection capability, estimation accuracy, and computational efficiency. The algorithm, augmented with the Gaussian back substitution technique, is assessed in both non-parallel and parallel computing environments to showcase its robustness and scalability. To this end, the P-ADMM algorithm, combined with the Gaussian back substitution technique, is applied to various problems, including Penalized Quantile Regression (PQR, \cite{Fan2021Penalized}) and support vector machines (SVM) with sparse regularization (\cite{Guan2020An}).

To select the optimal values for the regularization parameters $\bm \lambda$, we employ the method proposed in \cite{Zhang2016Variable} and \cite{Yu2017A}, which minimizes the HBIC criterion. The HBIC criterion is defined as follows:
\begin{align}
\text{HBIC}(\bm \lambda) = \log\left(\sum_{i=1}^{n} \mathcal{L}(y_i - \bm x_i^\top \hat{\bm \beta}_{\bm \lambda})\right) + |S_{\bm \lambda}| \frac{\log(\log n)}{n} C_n.
\end{align}
Here, $\mathcal{L}$ represents the specific loss function, $\hat{\bm \beta}_{\bm \lambda}$ denotes the parameter estimates obtained from the non-convex estimation, and $|S_{\bm \lambda}|$ denotes the number of non-zero coordinates in $\hat{\bm \beta}_{\bm \lambda}$. The value of $C_n = 6 \log (p)$ is recommended by \cite{Peng2015An} and \cite{Fan2021Penalized}.
By minimizing the HBIC criterion, we can effectively select the optimal values for $\bm \lambda$  for convex and non-convex estimation. This selection method balances model complexity and goodness of fit for optimal estimation.
Moreover, reviewing the correction steps \eqref{cr} and \eqref{cr22}, an additional parameter \(\nu\) needs to be selected. According to the proof of Theorem \ref{TH1}, as long as \(\nu\) is within the interval (0, 1), it guarantees the convergence and the linear convergence rate of the two proposed algorithms. Empirically, we suggest setting \(\nu = 0.75\), as this choice assigns a weight of 0.75 to the current iteration variable and a weight of 0.25 to the previous iteration variable. This selection places primary emphasis on the current iteration solution while also appropriately considering the previous iteration solution.

For the initial vector in the iterative algorithm, all elements are initialized to 0.01, although other values are also feasible. For all tested ADMM algorithms, the maximum number of iterations was set to 500, and the following stopping criterion was employed:
\begin{align*}
\frac{\|\bm \beta^k - \bm \beta^{k-1} \|_2}{\max (1, \|\bm \beta^k \|_2)} \le 10^{-4}.
\end{align*}
This stopping criterion ensures that the difference between the estimated coefficients in consecutive iterations does not exceed the specified threshold.
All experiments were conducted on a computer equipped with an AMD Ryzen 9 7950X 16-core processor (clocked at 4.50 GHz) and 32 GB of memory, using the R programming language. For clarity, we denote QPADM-slack(GB) as the standard QRADM-slack with Gaussian back substitution  (Algorithm \ref{alg1}), and M-QPADM-slack(GB) as the modified QRADM-slack with Gaussian back substitution (Algorithm \ref{alg3}).

\subsection{Synthetic Data}\label{sec451}
In the first simulation, the P-ADMM algorithm with Gaussian back substitution proposed in this section is used to solve the $\ell_1$ quantile regression ($\ell_1$-QR, \cite{Li2008L1}) problem, and its performance is compared with several state-of-the-art algorithms, including QPADM proposed by \cite{Yu2017A} and QPADM-slack proposed by \cite{Fan2021Penalized}. We have included the experiment on nonconvex regularization terms (SCAD and MCP) in the Appendix \ref{4A1}.

Although \cite{Yu2017A} and \cite{Fan2021Penalized} provide R packages for their respective algorithms, these packages are only compatible with the Mac operating system. Furthermore, the package in \cite{Yu2017A} only provides the estimated coefficients, lacking information such as the number of iterations and computational time. To ensure a fair comparison, the R code for the QPADM and QPADM-slack algorithms was rewritten based on the descriptions in their respective papers.

In the simulation study of this section, the models used in the numerical experiments of \cite{Peng2015An}, \cite{Yu2017A}, and \cite{Fan2021Penalized} were adopted. Specifically, data were generated from the heteroscedastic regression model $y = x_6 + x_{12} + x_{15} + x_{20} + 0.7 x_1 \epsilon$, where $\epsilon \sim N(0,1)$. The independent variables $(x_1,x_2,\dots,x_p)$ were generated in two steps.

\begin{itemize}
\item First, $\bm \tilde{x} = (\tilde{x}_1, \tilde{x}_2, \dots, \tilde{x}_p)^\top$ was generated from a $p$-dimensional multivariate normal distribution $N(\bm 0, \bm \Sigma)$, where the covariance matrix $\bm \Sigma$ satisfies $\Sigma_{ij} = 0.5^{|i-j|}$ for $1 \le i,j \le p$.

\item Second, $x_1$ was set to $\Phi(\tilde x_1)$, while for $j=2,\dots,p$, we directly set $x_j = \tilde x_j$.
\end{itemize}

As stated by \cite{Yu2017A} and \cite{Fan2021Penalized}, \(x_1\) does not present when \(\tau = 0.5\). For simplicity, in the ensuing numerical experiments pertaining to regression, the default selection for $\tau$ is set to 0.7. In a non-parallel environment ($M=1$), this section simulates datasets of different sizes, specifically ($n$, $p$) = (30,000, 1,000), (1,000, 30,000), (10,000, 30,000), and (30,000, 30,000). In a parallel environment ($M \ge 2$), datasets of sizes ($n$, $p$) = (200,000, 500) and (500,000, 1,000) are simulated. For each setting, 100 independent simulations were conducted. The average results for non-parallel and parallel computations are presented in Tables \ref{tab41} and \ref{tab42}, respectively.

\begin{table}[H]
\caption{Comparison results of ADMM algorithms under different data scales with LASSO penalty}\label{tab41}
\renewcommand{\arraystretch}{1.5}
\resizebox{1\linewidth}{!}{
\begin{threeparttable}
\begin{tabular}{lcccccccccc}
\hline
$(n,p)$ &  \multicolumn{5}{c}{(30000,1000)}  &\multicolumn{5}{c}{(10000,30000)}  \\ 
\cmidrule(lr){2-6}\cmidrule(lr){7-11}
 &  P1 &  P2 &  AE   &  Ite   &  Time  &  P1 &  P2 &  AE   &  Ite   &  Time \\ \hline
QPADM        & 100 & 100 & \textbf{0.018(0.03)} & 107.9(12.6)& 41.35(6.78) & 86 & 100 & 4.010(0.71) & 500+(0.0)  & 2015.5(69.2)  \\
QPADMslack   & 100 & 100 & 0.058(0.06) & 56.7(6.04)& 37.56(5.27) & 84 & 100 & 4.196(1.05) & 279(28.6) & 1776.5(34.2)  \\
QPADM-slack(GB)      & 100 & 100 & 0.029(0.04) & 47.3(5.11)  & 35.17(5.30) & 90 & 100 & 3.616(0.86) & 233(20.8)  & 1031.9(23.1) \\
M-QPADM-slack(GB)       & 100 & 100 & 0.021(0.03) & \textbf{31.6(4.50)} & \textbf{24.42(3.08)} & \textbf{97} & 100 & \textbf{3.241(0.78)} & \textbf{189(16.6)} & \textbf{795.6(16.3)} \\
\midrule[1pt]
 $(n,p)$ & \multicolumn{5}{c}{(1000,30000)} & \multicolumn{5}{c}{(30000,30000)} \\ 
\cmidrule(lr){2-6}\cmidrule(lr){7-11}
 &  P1 &  P2 &  AE   &  Ite     &  Time  &  P1 &  P2 &  AE   &  Ite     &  Time\\ \hline
QPADM  & 76 & 100 & 8.012(1.05) & 500+(0.0)  & 2024.6(60.1) &100 & 100 & 1.701(0.33) & 500+(0.0)   & 3217.3(68.6)\\
QPADMslack  & 73 & 100 & 8.324(1.22) & 243(27.8) & 1526.7(46.9) &100 & 100 & 2.010(0.39) & 322(21.8)  & 2713.6(41.7) \\
QPADM-slack(GB)   &85& 100 & 8.107(1.09) & 213(24.3)  & 928.1(20.8) &100 & 100 & 1.929(0.36) & 299(18.8)   & 1928.3(37.5)\\
M-QPADM-slack(GB)   &\textbf{91}& 100 & \textbf{7.352(0.77)} & \textbf{166(11.9)} & \textbf{833.6(15.7)} &100 & 100 & \textbf{1.512(0.09)} &\textbf{ 227(14.7)}  & \textbf{1532.6(28.1)} \\ \hline
\end{tabular}
\begin{tablenotes}
        \footnotesize
        \item[*] The symbols in this table are defined as follows: P1 (\%) represents the proportion of times $x_1$ is selected; P2 (\%) represents the proportion of times $x_6$, $x_{12}$, $x_{15}$, and $x_{20}$ are selected; AE denotes the absolute estimation error; Ite indicates the number of iterations; and Time (s) refers to the running time. The numbers in parentheses represent the corresponding standard deviations, and the optimal results are highlighted in bold.
\end{tablenotes}
\end{threeparttable}}
\end{table}

The results in Table \ref{tab41} reveal that the Gaussian back substitution method significantly enhances the convergence rate of QPADM-slack, as evidenced by the reduction in the number of iterations (Ite). The QPADM-slack(GB) algorithm introduces an additional correction step compared to QPADM-slack, which involves matrix-vector multiplication. Consequently, despite achieving a notably lower iteration count, QPADM-slack(GB) does not exhibit a substantial advantage in computational time (Time, in seconds) over QPADM-slack. This is because during the correction step, there are frequent matrix multiplication operations by vectors. By contrast, the modified M-QPADM-slack(GB), which adjusts the sequence of variable updates, demonstrates statistically significant improvements in both iteration count (Ite) and computational time (Time). %
In terms of computational precision, specifically absolute estimation error (AE), QPADM performs the best when the dimensionality is small. However, as the dimensionality increases, the situation changes. While the AE of QPADM remains better than that of QPADM-slack, it performs slightly worse compared to QPADM-slack(GB) and M-QPADM-slack(GB). Furthermore, regarding variable selection performance (P1 and P2), all methods perform well when the sample size ($n$) is larger than or slightly smaller than the dimensionality ($p$). However, in scenarios where dimensionality significantly exceeds sample size, the Gaussian back substitution methods exhibit statistically superior performance compared to their counterparts.

\begin{figure}[H]
    \centering
    \subfigure[Estimation errors]{
    \label{Fig.sub.1}
    \includegraphics[width=0.45\linewidth]{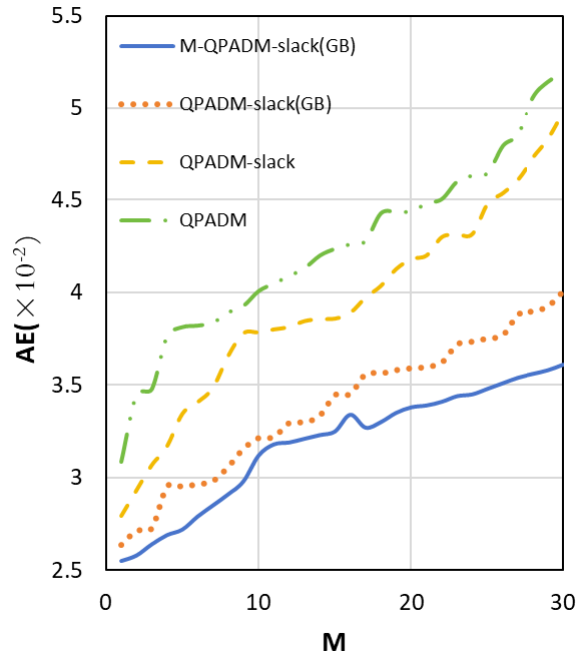}}
    \subfigure[Computation times]{
    \label{Fig.sub.2}
    \includegraphics[width=0.45\linewidth]{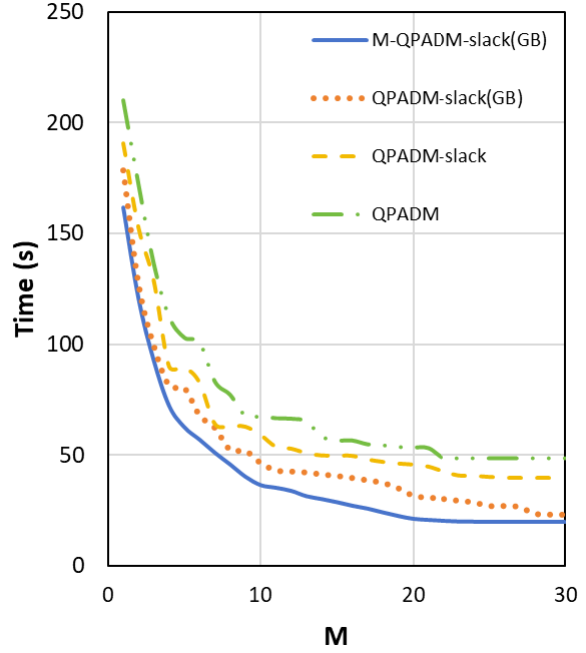}}
    \caption{Visualization of estimation errors and computation times for various P-ADMM algorithms as $M$ increases.}
    \label{fig43}
\end{figure}

In Table \ref{tab42}, as the number of local sub-machines ($M$) increases, a deteriorating trend is observed in both Nonzero and AE for all variants of P-ADMM, a numerical observation consistent with the findings in the numerical experiments reported by \cite{Fan2021Penalized}. Across different numbers of local sub-machines ($M$), M-QPADM-slack(GB) consistently surpasses other P-ADMM methods by a notable margin, exhibiting superior performance in both computational accuracy and efficiency.
This further supports the necessity of modifying the iteration sequence in QPADM-slack when Gaussian back substitution is incorporated.   

More detailed results on estimation errors and computation times are presented in Figure \ref{fig43}. 
Concerning computation time, we noticed that after $M$ surpasses 20, the tendency of computation time to decrease with the increase of $M$ diminishes and may stabilize. This is not a result of the algorithm struggling to manage a high number of sub-machines, but rather a constraint imposed by our computer's configuration. Using a machine with ample memory would prevent this issue from arising.

\begin{table}[H]
\caption{Comparison of different P-ADMM algorithms under the LASSO penalty}\label{tab42}
\renewcommand{\arraystretch}{1.5}
\resizebox{1\columnwidth}{!}{
\begin{threeparttable}
\begin{tabular}{lcccccccc}
\hline
  & \multicolumn{2}{l}{QPADM} & \multicolumn{2}{r}{$(200000,500)$} & \multicolumn{2}{l}{QPADM} & \multicolumn{2}{r}{$(500000,1000)$}\\ 
\cmidrule(lr){2-5}\cmidrule(lr){6-9}
M & Nonzero &  AE   &  Ite   &  Time  & Nonzero &  AE   &  Ite   &  Time   \\ \hline
 5  & 41.0(3.83) & {0.074(0.0009)} & 359.4(27.1) & 80.1(5.82) & 28.3(2.15) & 0.042(0.0006) & 442.1(37.0) & 177.2(12.6)  \\
10  & 44.5(4.01) & 0.071(0.0009) & 372.3(28.9) & 48.2(2.98) & 29.1(2.33) & 0.049(0.0007) & 471.2(40.8) & 87.6(6.63) \\
20 & 47.2(4.32) & 0.075(0.0011) & 405.2(32.6) & 29.7(1.64) & 32.2(2.01) & 0.052(0.0008) & 494.5(47.1) & 43.5(3.88)  \\\hline
   & \multicolumn{2}{l}{QPADM-slack}&\multicolumn{2}{r}{$(200000,500)$} & \multicolumn{2}{l}{QPADM-slack} & \multicolumn{2}{r}{$(500000,1000)$} \\  \cmidrule(lr){2-5}\cmidrule(lr){6-9}
M  &  Nonzero &  AE           &  Ite        &  Time & Nonzero & AE & Ite & Time \\\hline
 5  & 36.5(2.95) & 0.079(0.0010) & 255.6(22.0) & 61.3(3.56) & 25.2(1.92) & 0.049(0.0007) & 361.5(32.6) & 136.6(9.32)  \\
 10 & 39.9(3.06) & 0.080(0.0011) & 269.1(26.1) & 35.8(2.28) & 28.9(2.09) & 0.051(0.0008) & 379.9(36.8) & 78.4(5.17)   \\
20 & 42.3(3.16) & 0.083(0.0013) & 356.8(42.2) & {22.6(1.53)} & 31.6(2.23) & 0.055(0.0009) & 423.6(40.2) & 42.9(2.61)  \\\hline
   & \multicolumn{2}{l}{QPADM-slack(GB) } & \multicolumn{2}{r}{$(200000,500)$} & \multicolumn{2}{l}{QPADM-slack(GB)} & \multicolumn{2}{r}{$(500000,1000)$}\\   \cmidrule(lr){2-5}\cmidrule(lr){6-9}
M  &  Nonzero &  AE           &  Ite        &  Time  & Nonzero & AE & Ite & Time \\\hline
 5  & {25.5(2.20)} & 0.058(0.0006) & {195.6(12.6)}   & {49.5(3.03)} & {24.4(1.90)} & {0.033(0.0005)} & {258.1(23.2)} & {82.6(6.62)}    \\
 10 & {26.4(2.41)} & {0.062(0.0008)} &{203.0(13.4)}  & {29.8(1.63)}  & {25.6(1.97)} & {0.037(0.0005)} & {269.3(23.5)} & {50.2(3.88)}  \\
20 & {27.0(2.43)} & {0.065(0.0009)} & {211.9(13.5)}   & {15.7(0.92)}  & {26.2(2.04)} & {0.039(0.0006)} & {271.0(26.0)} & {35.1(2.12)}  \\\hline
 & \multicolumn{2}{l}{M-QPADM-slack(GB)  } & \multicolumn{2}{r}{$(200000,500)$} & \multicolumn{2}{l}{M-QPADM-slack(GB)} & \multicolumn{2}{r}{$(500000,1000)$}\\  \cmidrule(lr){2-5}\cmidrule(lr){6-9}
M  &  Nonzero &  AE           &  Ite        &  Time  & Nonzero & AE & Ite & Time  \\\hline
 5  & \bf{20.1(1.91)} & \bf{0.050(0.0005)} & \bf{148.5(9.31)}   & \bf{39.2(2.82)}  & \bf{15.21(1.44)} & \bf{0.027(0.0004)} & \bf{196.7(13.6)} & \bf{62.5(4.32)}  \\
 10 & \bf{21.6(1.97)} & \bf{0.053(0.0005)} &\bf{152.1(9.92)}  & \bf{22.2(1.73)}  & \bf{15.38(1.57)} & \bf{0.030(0.0005)} & \bf{199.6(13.4)} & \bf{36.4(2.73)}  \\
20 & \bf{22.3(2.01)} & \bf{0.054(0.0006)} & \bf{156.6(10.8)}   & \bf{12.9(0.95)}  & \bf{15.52(1.64)} & \bf{0.033(0.0005)} & \bf{201.1(13.9)} & \bf{20.2(1.51)}  \\\hline
\end{tabular}
\begin{tablenotes}
        \footnotesize
        \item[*] Since the values of all methods for metrics P1 and P2 are 100, they are not listed in Table \ref{tab42}. ``Nonzero" indicates the number of non-zero coefficients in the estimates. The numbers in parentheses represent the corresponding standard deviations, and the optimal results are shown in bold.
\end{tablenotes}
\end{threeparttable}}
\end{table}

\subsection{Real Data Experiment}\label{sec452}
In this section, the empirical analysis focuses on classification tasks using real-world data. The dataset employed is rcv1.binary, which consists of 47,236 features, 20,242 training samples, and 677,399 testing samples. This dataset is publicly accessible at \url{https://www.csie.ntu.edu.tw/~cjlin/libsvmtools/datasets/binary.html#rcv1.binary}.

In the subsequent experiments, the training samples are utilized to fit the model, where the data matrix $\bm X$ has dimensions $n=20,242$ (number of observations) and $p=47,236$ (number of features). This high-dimensional setting, where $p > n$, necessitates the use of regularization techniques to address potential overfitting. Specifically, an $\ell_1$ regularization term is incorporated to induce sparsity in the model, as many features are expected to be irrelevant for classification. The experiment on SVM with nonconvex regularization terms is included in Appendix \ref{4A2}.

As demonstrated in Proposition 1 of \cite{Wu2025Multi}, the optimization algorithms for piecewise linear classification and regression models are interchangeable. Consequently, both QPADM and QPADM-slack can be applied to solve the $\ell_1$-SVM problem (\cite{Zhu20031-norm}). Notably, the algorithm proposed in \cite{Guan2020An} reformulates the piecewise linear classification loss using slack variables, making it equivalent to QPADM-slack with the quantile parameter $\tau = 1$. Therefore, QPADM-slack ($\tau = 1$) is adopted as the baseline method for comparison in the following analysis.

To evaluate the algorithms, several performance metrics are defined. For the testing set, a random subsample of 10,000 observations is drawn from the 677,399 testing samples, yielding $n_{\text{test}}=10,000$. The metrics include: (1) ``Time", which measures the computational runtime of the algorithm; (2) ``Iteration", which records the number of iterations required for convergence; (3) ``Sparsity", defined as the proportion of zero coefficients (i.e., the number of zero coefficients divided by the total number of coefficients); (4) ``Train", which quantifies the classification accuracy on the training set; and (5) ``Test", which measures the classification accuracy on the testing set. Each experimental setting is independently simulated 100 times, and the average results are reported in Table \ref{tab43}. Since the previous experiments demonstrated that the Gaussian back substitution technique outperforms QPADM-slack in terms of iteration count and computational time, the results for ``Ite" and ``Time" are omitted in Table \ref{tab43}.

\begin{table}[!ht]\footnotesize
    \centering
    \renewcommand{\arraystretch}{1.5}
    \caption{Comparative analysis of Sparsity, training, and average testing accuracies (\%) for $\ell_1$ SVM using M-QPADM-slack(GB), QPADM-slack, and QPADM-slack(GB)}
    \begin{tabular}{llllllllll}
    \hline
          & \multicolumn{3}{c}{M-QPADM-slack(GB)} & \multicolumn{3}{c}{QPADM-slack(GB)} & \multicolumn{3}{c}{QPADM-slack}\\ 
        \cmidrule(lr){2-4}\cmidrule(lr){5-7}\cmidrule(lr){8-10}
        M & Sparsity & Train & Test & Sparsity & Train & Test & Sparsity & Train & Test \\ \hline
        2 & \bf{90.36} & \bf{99.42} & \bf{97.25} & 85.98 & 95.12 & 92.98 & 83.63 & 93.37 & 91.99 \\ 
        4 & \bf{90.36} & \bf{99.23} & \bf{97.15} & 84.60 & 94.55 & 92.02 & 81.36 & 92.25 & 91.04 \\ 
        6 & \bf{90.36} & \bf{99.15} & \bf{97.05} & 83.59 & 93.87 & 91.46 & 80.13 & 91.76 & 90.58 \\ 
        8 & \bf{90.31} & \bf{99.02} & \bf{96.91} & 82.10 & 93.03 & 90.93 & 79.36 & 91.04 & 89.91 \\ 
        10 & \bf{90.12} & \bf{98.89} & \bf{96.82} & 81.36 & 92.56 & 90.07 & 77.98 & 90.76 & 89.07 \\ 
        12 & \bf{90.01} & \bf{98.77} & \bf{96.69} & 80.75 & 91.87 & 89.38 & 77.12 & 90.01 & 88.36 \\ 
        14 & \bf{89.85} & \bf{98.65} & \bf{96.57} & 79.36 & 91.15 & 88.64 & 76.42 & 89.72 & 87.82 \\ 
        16 & \bf{89.58} & \bf{98.51} & \bf{96.41} & 78.69 & 90.36 & 87.39 & 73.95 & 88.36 & 87.05 \\ 
        18 & \bf{89.65} & \bf{98.43} & \bf{96.36} & 77.65 & 89.78 & 86.27 & 72.65 & 87.77 & 86.23 \\ 
        20 & \bf{89.55} & \bf{98.32} & \bf{96.27} & 76.25 & 88.62 & 85.33 & 71.45 & 87.65 & 85.16 \\ \hline
    \end{tabular}
    \label{tab43}
\end{table}

The numerical results in Table \ref{tab43} reveal that, not only in the regression numerical experiments, but also in the classification experiments, the performance of the three P-ADMM methods based on the consensus structure deteriorates as the number of local sub-machines ($M$) increases. This occurs because the consensus structure, designed to support the parallel algorithm, unavoidably introduces consensus constraints, namely auxiliary variables, as the number of local sub-machines grows. These additional auxiliary variables degrade the quality of the iterative solutions of the P-ADMM algorithm. In terms of classification accuracy (Train and Test) and variable selection (Sparsity), M-QPADM-slack(GB) performs the best, followed by QPADM-slack(GB), and lastly, QPADM-slack. This indicates that when solving  sparse regularization classification problems  using the slack-based P-ADMM algorithm, incorporating the Gaussian  back substitution technique can enhance both the efficiency and accuracy of the algorithm.

\section{Conclusions and Future Prospects}\label{sec46}

This paper introduces a Gaussian back substitution technique to adapt the parallel ADMM (P-ADMM) algorithms proposed by \cite{Guan2020An} and \cite{Fan2021Penalized}. Specifically, we incorporate minor linear adjustments at each iteration step, which, despite their minimal computational overhead, significantly enhance the algorithm's convergence speed. More importantly, this technique theoretically achieves linear convergence rates and demonstrates high efficiency and robustness in practical applications.

Furthermore, this paper proposes a new iteration variable sequence for P-ADMM with slack variables. When combined with the Gaussian back-substitution technique, this approach can significantly enhance computational efficiency. The new iteration sequence is designed to better align with the Gaussian back-substitution technique. By streamlining key steps in the algorithm to involve only basic addition and subtraction operations, we avoid complex matrix-vector multiplications, thereby improving overall efficiency. This is particularly advantageous for large datasets and high-dimensional spaces, as it avoids computational load and complexity. This modification not only preserves the algorithm's convergence properties but potentially accelerates the convergence process, offering a more efficient solution for practical applications.

Another significant contribution of this paper is the extension of quantile loss applications from traditional regression tasks to classification tasks. Given the intrinsic equivalence of quantile loss optimization forms in both classification and regression tasks, the proposed parallel algorithm can be seamlessly applied to quantile regression classification models. This extension provides a novel perspective and tool for addressing classification problems.

Despite the advancements presented in this paper, a challenges remains. To be specific, even with the incorporation of the Gaussian back-substitution technique, the quality of the solution and the algorithm's convergence speed deteriorate as the number of local sub-machines increases. This limitation arises from the consensus-based parallel structure, which complicates resource allocation in practical applications. Consequently, developing a parallel algorithm that circumvents the consensus structure emerges as a highly promising research direction. Such an algorithm could provide more robust and efficient solutions for larger datasets and higher-dimensional feature spaces.

\singlespacing
\section*{Acknowledgements}
We express our sincere gratitude to Professor Bingsheng He for engaging in invaluable discussions with us. His insights and expertise have greatly assisted us in effectively utilizing Gaussian back substitution to correct the convergence issues of the QPADM-slack algorithm. The research of Zhimin Zhang was supported by the National Natural Science Foundation of China [Grant Numbers 12271066, 12171405], and the research of Xiaofei Wu  was supported by the Scientific and Technological Research Program of Chongqing Municipal Education Commission [Grant Numbers KJQN202302003].

\footnotesize{
\bibliographystyle{elsarticle-harv}
\bibliography{myrefq}}

\onehalfspacing
\newpage
\section*{Online Appendix}
\appendix

\section{Proofs of Convergence Theorems}\label{4B}
\subsection{Preliminary}

\subsubsection{Lemma}
\begin{lem}(Lemma 2.1 in \cite{He2022A}) \label{lem41}
Let \( \mathbb{Z} \subset \mathbb{R}^l \) be a closed convex set, and let \( \theta : \mathbb{R}^l \to \mathbb{R} \) and \( f : \mathbb{R}^l \to \mathbb{R} \) be convex functions. Suppose \( f \) is differentiable on an open set containing \( \mathbb{Z} \), and the minimization problem
\[
\min \{ \theta(z) + f(z) \mid z \in \mathbb{Z} \}
\]
has a nonempty solution set. Then, \( z^* \in \arg \min \{ \theta(z) + f(z) \mid z \in \mathbb{Z} \} \) if and only if
\[
z^* \in \mathbb{Z} \quad \text{and} \quad \theta(z) - \theta(z^*) + (z - z^*)^\top \nabla f(z^*) \geq 0, \quad \forall z \in \mathbb{Z}.
\]
\end{lem}

\begin{lem}\label{lem42}
Assume that \( \bm{G} \) and \( \bm{H} \in \mathbb{R}^{n \times n} \) are both positive definite matrices, and \( \bm{a}, \bm{b}, \bm{c}, \bm{d} \) are four arbitrary \( n \)-dimensional vectors. Then the following identities hold:
\begin{enumerate}
    \item \begin{align}\label{Id1}
    \| \bm{a} \|_{\bm{H}}^2 - \| \bm{b} \|_{\bm{H}}^2 = 2 \bm{a}^\top \bm{H} (\bm{a} - \bm{b}) - \| \bm{a} - \bm{b} \|_{\bm{H}}^2.
\end{align}
    \item \begin{align}\label{Id2}
    2 (\bm{a} - \bm{b})^\top \bm{H} (\bm{c} - \bm{d}) = (\| \bm{a} - \bm{d} \|_{\bm{H}}^2 - \| \bm{a} - \bm{c} \|_{\bm{H}}^2) + (\| \bm{c} - \bm{b} \|_{\bm{H}}^2 - \| \bm{d} - \bm{b} \|_{\bm{H}}^2).
 \end{align}
\end{enumerate}
\end{lem}
The conclusion of Lemma \ref{lem42} is straightforward to verify. Despite its simplicity, this result is widely utilized in the proof of convergence for ADMM, as detailed in \cite{He2018A} and \cite{He2022A}.

\subsubsection{Four Matrices}
To simplify the presentation of analysis, we  define the following four matrices,
\small{\begin{align}\label{m1}
\bm{Q} = \begin{bmatrix}
    \mu \bm{B}^\top \bm{B} & \bm{0} & \bm{0} \\
    \mu \bm{C}^\top \bm{B} & \mu \bm{C}^\top \bm{C} & \bm{0} \\
    - \bm{B} & - \bm{C} & \frac{\bm{I}}{\mu}
\end{bmatrix},
\
\bm{M} = \begin{bmatrix}
    \nu \bm{I} & - \nu (\bm{B}^\top \bm{B})^{-1} \bm{B}^\top \bm{C} & \bm{0} \\
    \bm{0} & \nu \bm{I} & \bm{0} \\
    - \mu \bm{B} & - \mu \bm{C} & \bm{I}
\end{bmatrix},
\end{align}}

\small{\begin{align}\label{m2}
\bm{H} = \begin{bmatrix}
    \frac{\mu}{\nu} \bm{B}^\top \bm{B} & \frac{\mu}{\nu} \bm{B}^\top \bm{C} & \bm{0} \\
    \frac{\mu}{\nu} \bm{C}^\top \bm{B} & \frac{\mu}{\nu} \left[ \bm{C}^\top \bm{C} + \bm{C}^\top \bm{B} (\bm{B}^\top \bm{B})^{-1} \bm{B}^\top \bm{C} \right] & \bm{0} \\
    \bm{0} & \bm{0} & \frac{\bm{I}}{\mu}
\end{bmatrix},
\bm{G} = \begin{bmatrix}
    (1 - \nu) \mu \bm{B}^\top \bm{B} & \bm{0} & \bm{0} \\
    \bm{0} & (1 - \nu) \mu \bm{C}^\top \bm{C} & \bm{0} \\
    \bm{0} & \bm{0} & \frac{\bm{I}}{\mu}
\end{bmatrix}.
\end{align}}
It is obvious that the matrices $\bm Q$, $\bm M$, $\bm H$ and $\bm G$ are all $3  \times 3$ partitioned, and they satisfy the condition  
\begin{align}\label{eq1}
\bm H \bm M = \bm Q  \  \text{and} \
 \bm G = (\bm Q^\top + \bm Q) -  \bm M^\top \bm H \bm M.
\end{align}
 It is not difficult to verify that both  $\bm H$ and  $\bm G$ are positive-definite matrices when $ \nu \in (0,1)$.

\subsubsection{Variational Inequality Characterization}
To transform the optimization objective function (\ref{intr3}) into an equality optimization problem, that is, to eliminate the nonnegative constraints of $\bm \xi_m \ge \bm 0$ and $\bm \eta_m \ge \bm 0$, we need to introduce the following  indicator function,
\[
I_{+}(\bm{x}) = \begin{cases}
0, & \text{if } \bm{x} \in \mathbb{R}^+, \\
\infty, & \text{if } \bm{x} \notin  \mathbb{R}^+.
\end{cases}
\]

If we consider the QPADM-Slack(GB) method, we define:
\[
\bm{a} = (\bm{\beta}, \bm{\xi}_1, \dots, \bm{\xi}_M), \quad \theta_1(\bm{a}) = P_\lambda(\bm{\beta}) + \sum_{m=1}^{M} \left[ \tau \bm{1}_{n_m}^\top \bm{\xi}_m + I_{+}(\bm{\xi}_m) \right]
\]
\[
\bm{b} = (\bm{\eta}_1, \dots, \bm{\eta}_M), \quad \theta_2(\bm{b}) = \sum_{m=1}^{M} \left[ (1 - \tau) \bm{1}_{n_m}^\top \bm{\eta}_m + I_{+}(\bm{\eta}_m) \right]
\]
\[
\bm{c} = (\bm{\beta}_1, \dots, \bm{\beta}_M), \quad \theta_3(\bm{c}) = 0
\]
\[
\bm{e} = (\bm{0}, \dots, \bm{0}, \bm{y}_1^\top, \dots, \bm{y}_M^\top)^\top
\]
Thus, from (\ref{conm}), equation (\ref{intr3}) can be rewritten as:
\begin{equation}
\begin{aligned}\label{mslack}
\min_{\bm{a}, \bm{b}, \bm{c}} \quad & \theta_1(\bm{a}) + \theta_2(\bm{b}) + \theta_3(\bm{c}), \\
\text{s.t.} \quad & \bm{A}\bm{a} + \bm{B}\bm{b} + \bm{C}\bm{c} = \bm{e},
\end{aligned}
\end{equation}
where the definitions of \(\bm{A}\), \(\bm{B}\), \(\bm{C}\), and \(\bm{e}\) are provided in Section \ref{sec31}.

If the modified update sequence as described in Section \ref{sec41} is employed, the variables are defined as follows:
\[
\bm{a} = (\bm{\beta}_1, \dots, \bm{\beta}_M), \quad \bm{b} = (\bm{\xi}_1, \dots, \bm{\xi}_M), \quad \bm{c} = (\bm{\eta}_1, \dots, \bm{\eta}_M, \bm{\beta}).
\]
Equation (\ref{intr3}) can then be rewritten as (\ref{mslack}), with the corresponding matrices \(\bm{A}\), \(\bm{B}\), and \(\bm{C}\) as detailed in Section \ref{sec41}.

Based on the variational inequality characterization discussed in Section 2.4.2 of this paper, the solution to the constrained optimization problem (\ref{mslack}) is the saddle point of the following Lagrangian function:
\begin{align}\label{la}
L(\bm{a}, \bm{b}, \bm{c}, \bm{d}) = \theta_1(\bm{a}) + \theta_2(\bm{b}) + \theta_3(\bm{c}) - \bm{d}^\top (\bm{Aa} + \bm{Bb} + \bm{Cc} - \bm{e}),
\end{align}
where \(\bm{d} = (\bm{d}_1, \dots, \bm{d}_M, \bm{e}_1, \dots, \bm{e}_M)\). However, since \(\theta_1\) and \(\theta_2\) are non-differentiable in this context, the variational inequality mentioned in \cite{He2018A} cannot be directly applied to our scenario. Recently, the work in \cite{He2022A} has covered the non-differentiable setting considered in this paper. As described in their Section 2.2, finding the saddle point of \(L(\bm{a}, \bm{b}, \bm{c}, \bm{d})\) is equivalent to finding \(\bm{a}^*, \bm{b}^*, \bm{c}^*, \bm{d}^*\) that satisfy the following variational inequality:
\begin{equation}\label{vi}
\bm{h}^* \in \Omega, \quad \theta(\bm{f}) - \theta(\bm{f}^*) + (\bm{h} - \bm{h}^*)^\top F(\bm{h}^*) \geq 0, \quad \forall \bm{h} \in \Omega,
\end{equation}
where \(\bm{f} = (\bm{a}, \bm{b}, \bm{c})\), \(\bm{h} = (\bm{a}, \bm{b}, \bm{c}, \bm{d})\), and
\begin{equation}
\theta(\bm{f}) = \theta_1(\bm{a}) + \theta_2(\bm{b}) + \theta_3(\bm{c}), \quad \Omega = \mathbb{R}^{(p + Mn)} \times \mathbb{R}^{(Mn)} \times \mathbb{R}^{(Mn)} \times \mathbb{R}^{(2M)},
\end{equation}
and
\begin{equation}\label{F}
F(\bm{h}) = \begin{pmatrix}
-\bm{A}^\top \bm{d} \\
-\bm{B}^\top \bm{d} \\
-\bm{C}^\top \bm{d} \\
\bm{Aa} + \bm{Bb} + \bm{Cc} - \bm{e}
\end{pmatrix}.
\end{equation}

It is important to note that the operator \(F\) defined in (\ref{F}) is antisymmetric , as
\begin{equation}\label{f}
(\bm{h}_1 - \bm{h}_2)^\top \left[F(\bm{h}_1) - F(\bm{h}_2)\right] \equiv 0, \quad \forall \ \bm{h}_1, \bm{h}_2 \in \Omega.
\end{equation}

In this chapter, \(\Omega^*\) denotes the solution set of (\ref{vi}), which is also the set of saddle points of the Lagrangian function (\ref{la}) for the model (\ref{intr3}).

\subsection{Proof}
 With the preparation above, we can now proceed to prove the convergence and convergence rate of QPADM-Slack with Gaussian backtracking (which includes both QPADM-Slack(GB) and M-QPADM-Slack(GB)). For brevity, both of these methods with Gaussian backtracking will hereafter be referred to as QPADM-Slack(GB).

The process of QPADM-Slack(GB) can be divided into two steps: The first step involves generating the predicted values \( \tilde{\bm a}^k \), \( \tilde{\bm b}^k \), \( \tilde{\bm c}^k \), and \( \tilde{\bm d}^k \); the second step involves correcting \( \tilde{\bm b}^k \), \( \tilde{\bm c}^k \), and \( \tilde{\bm d}^k \) to produce \( {\bm b}^{k+1} \), \( {\bm c}^{k+1} \), and \( {\bm d}^{k+1} \).

The specific iteration formulas for the first step (prediction step) are as follows:
\begin{small}
\begin{equation}\label{4priupdate}
\left\{
\begin{array}{l}
\tilde{\boldsymbol{a}}^{k}  =  \mathop{\arg \min }\limits_{\boldsymbol a} \left\{  \theta_1(\bm a ) - \bm a^\top \bm  A^\top \bm d^k + \frac{\mu}{2}\|\bm {Aa} +  \bm {B} \bm b^k + \bm {C} \bm c^k - \bm e   \|^2  \right\}; \\
\tilde{\boldsymbol{b}}^{k}  =  \mathop{\arg \min }\limits_{\boldsymbol b} \left\{  \theta_2(\bm b ) - \bm b^\top \bm  B^\top \bm d^k + \frac{\mu}{2}\|\bm {A} \tilde{\bm a}^k +  \bm {B} \bm b + \bm {C} \bm c^k - \bm e   \|^2  \right\}; \\
\tilde{\boldsymbol{c}}^{k}  =  \mathop{\arg \min }\limits_{\boldsymbol c} \left\{  \theta_3(\bm c ) - \bm c^\top \bm  C^\top \bm d^k + \frac{\mu}{2}\|\bm {A} \tilde{\bm a}^k +  \bm {B} \tilde{\bm b}^k + \bm {C} \bm c - \bm e   \|^2  \right\};\\
\tilde{\bm d}^k = \bm d^k - \mu (\bm {A} \tilde{\bm a}^k +  \bm {B} {\bm b}^k + \bm {C}{\bm c}^k - \bm e  );
\end{array}\right.
\end{equation}
\end{small}
where \(\bm d^k = (\bm u_1^k, \dots, \bm u_M^k, \bm v_1^k, \dots, \bm v_M^k)\). Here, the updates for \( \tilde{\bm a}^k \), \( \tilde{\bm b}^k \), and \( \tilde{\bm c}^k \) are essentially the same as those for \( {\bm a}^{k+1} \), \( {\bm b}^{k+1} \), and \( {\bm c}^{k+1} \) in QPADM-Slack, or as improved in Section \ref{sec41}. For the sake of convenience in subsequent proofs, the update method for \( \tilde{\bm d}^k \) differs from that in QPADM-Slack, primarily because it does not utilize the newly generated \( \tilde{\boldsymbol{b}}^{k} \) and \( \tilde{\boldsymbol{c}}^{k} \).

According to Lemma \ref{lem41}, we have
\begin{small}
\begin{equation}\label{vi1}
\left\{
\begin{array}{l}
\theta_1(\bm a ) -  \theta_1(\tilde{\bm a}^k ) + (\bm a - \tilde{\bm a}^k )^\top \left[-\bm A^\top \bm d^k + \mu \bm A^\top ( \bm {A} \tilde{\bm a}^k +  \bm {B} \bm b^k + \bm {C} \bm c^k - \bm e )   \right]; \\
\theta_2(\bm b ) -  \theta_2(\tilde{\bm b}^k ) + (\bm b - \tilde{\bm b}^k )^\top \left[-\bm B^\top \bm d^k + \mu \bm B^\top ( \bm {A} \tilde{\bm a}^k +  \bm {B} \tilde{ \bm b}^k + \bm {C} \bm c^k - \bm e )    \right]; \\
\theta_3(\bm c ) -  \theta_3(\tilde{\bm c}^k ) + (\bm c - \tilde{\bm c}^k )^\top \left[-\bm C^\top \bm d^k + \mu \bm C^\top ( \bm {A} \tilde{\bm a}^k +  \bm {B} \tilde{ \bm b}^k + \bm {C} \tilde{\bm c}^k - \bm e )    \right]; \\
(\bm d - \tilde{\bm d}^k)^\top \left[ (\tilde{\bm d}^k - {\bm d}^k)/\mu + \bm {A} \tilde{\bm a}^k +  \bm {B} {\bm b}^k + \bm {C}{\bm c}^k - \bm e  \right] = 0.
\end{array}\right.
\end{equation}
\end{small}
The last equation is derived from the final part of equation (\ref{4priupdate}). By combining the four equations in (\ref{vi1}), we obtain
\begin{align}\label{vi2} 
\theta(\bm f) - \theta( \tilde{ \bm f}^k) + (\bm h - \tilde{ \bm h}^k)^T F ( \tilde{ \bm h}^k) \geq  (\bm g - \tilde{ \bm g}^k)^T \bm Q  (\bm g^k - \tilde{ \bm g}^k),  \quad \forall \bm h \in \Omega,
\end{align}
where \( \bm{g} = (\bm{b}, \bm{c}, \bm{d}) \), and by definition, \(\bm{f}\) and \(\bm{g}\) are both components of \(\bm{h}\).

 The second step  (correction step) involves setting \(\bm a^{k+1} = \tilde{\boldsymbol{a}}^{k}\), followed by the equations:
\begin{align} \label{vi3}
\begin{bmatrix}
   {\bm b}^{k+1}\\
 {\bm c}^{k+1} \\
{\bm d}^{k+1}
\end{bmatrix}  =  \begin{bmatrix}
   {\bm b}^{k}\\
 {\bm c}^{k} \\
{\bm d}^{k}
\end{bmatrix} -   \begin{bmatrix}
    \nu \bm I & - \nu (\bm B^\top \bm B)^{-1} \bm B^\top \bm C & \bm 0  \\
    \bm 0 & \nu \bm I & \bm 0 \\
    - \mu \bm B &  - \mu \bm C   & {\bm I} 
\end{bmatrix}  \begin{bmatrix}
   {\bm b}^{k} -   \tilde{\bm b}^k \\
 {\bm c}^{k}  -  \tilde{\bm c}^k\\
{\bm d}^{k} -  \tilde{\bm d}^k
\end{bmatrix}. 
\end{align}
The first two rows of the above equation correspond to the correction steps for \(\tilde{\bm \xi}^k_m\) and \(\tilde{\bm \eta}^k_m\), while the third row adjusts \(\tilde{\bm d}^k\) to align with the expression for \(\bm d^{k+1}\) in QPADM-Slack. The adjustment in the third row is necessary because QPADM-Slack (GB) only requires corrections for \(\tilde{\bm \xi}^k_m\) and \(\tilde{\bm \eta}^k_m\). Hence, the correction step in  (\ref{vi3}) can be rewritten as
\begin{align}\label{vi4}
\bm{g}^{k+1} = \bm{g}^{k} - \bm{M} (\bm{g}^{k} - \tilde{\bm{g}}^{k}),
\end{align}
where the definition of the matrix \(\bm M\) is given in (\ref{m1}).

The above discussion indicates that the QPADM-slack (GB) algorithm can be simplified into two steps: the prediction step (see  (\ref{4priupdate})) and the correction step (see  (\ref{vi3})). The prediction and correction steps can also be expressed as:
\begin{align}
\theta(\bm f) - \theta( \tilde{ \bm f}^k) +  (\bm h - \tilde{ \bm h}^k)^\top F ( \tilde{ \bm h}^k) \geq & \  (\bm g - \tilde{ \bm g}^k)^\top \bm Q  (\bm g^k - \tilde{ \bm g}^k),  \quad \forall \bm h \in \Omega, \label{p} \\
\bm{g}^{k+1} = &  \ \bm{g}^{k} - \bm{M} (\bm{g}^{k} - \tilde{\bm{g}}^{k}). \label{c}
\end{align}

\subsubsection{Global Convergence}
Next, based on the predictive and corrective steps, we derive that the sequence \(\{\bm{g}^k\}\) exhibits a contraction property.
\begin{prop}\label{prop41}
For the sequence \(\{\bm{g}^k\}\) generated by the QPADM-slack(GB) algorithm (including Algorithm \ref{alg1} and Algorithm \ref{alg3}), the following inequality holds:
\begin{align}\label{res1}
\| \bm{g}^{k+1} - \bm{g}^* \|_{\bm{H}}^2 \leq \| \bm{g}^k - \bm{g}^* \|_{\bm{H}}^2 - \| \bm{g}^k - \tilde{\bm{g}}^k \|_{\bm{G}}^2, \quad \forall \bm{g}^* \in \Omega^*,
\end{align}
where \(k > 0\) and \(\bm{G}\) and \(\bm{H}\) are defined in (\ref{m2}).
\end{prop}

\begin{proof}
By applying \(\bm{H} \bm{M} = \bm{Q}\) (see  (\ref{eq1})) and the update expression in  (\ref{c}), the right-hand side of (\ref{p}) can be rewritten as:
\begin{align}\label{pr1}
\left[ \theta(\bm{f}) - \theta(\tilde{\bm{f}}^k) + (\bm{h} - \tilde{\bm{h}}^k)^\top \bm{F} (\tilde{\bm{h}}^k) \right] \geq (\bm{g} - \tilde{\bm{g}}^k)^\top \bm{H} (\bm{g}^k - \bm{g}^{k+1}), \quad \forall \bm{h} \in \Omega.
\end{align}
Using the identity in  (\ref{Id1}), let \(\bm{a} = \bm{g}\), \(\bm{b} = \tilde{\bm{g}}^k\), \(\bm{c} = \bm{g}^k\), and \(\bm{d} = \bm{g}^{k+1}\), we obtain:
\begin{align}\label{pr2}
(\bm{g} - \tilde{\bm{g}}^k)^\top \bm{H} (\bm{g}^k - \bm{g}^{k+1}) &= \frac{1}{2} \left( \|\bm{g} - \bm{g}^{k+1}\|_{\bm{H}}^2 - \|\bm{g} - \bm{g}^k\|_{\bm{H}}^2 \right) \nonumber \\
&+ \frac{1}{2} \left( \|\bm{g}^k - \tilde{\bm{g}}^k\|_{\bm{H}}^2 - \|\bm{g}^{k+1} - \tilde{\bm{g}}^k\|_{\bm{H}}^2 \right).
\end{align}
For the last term in  (\ref{pr2}), we have:
\begin{equation}\label{pr3}
\begin{aligned}
\|\bm{g}^k - \tilde{\bm{g}}^k\|_{\bm{H}}^2 - \|\bm{g}^{k+1} - \tilde{\bm{g}}^k\|_{\bm{H}}^2 &= \|\bm{g}^k - \tilde{\bm{g}}^k\|_{\bm{H}}^2 - \|(\bm{g}^k - \tilde{\bm{g}}^k) - (\bm{g}^k - \bm{g}^{k+1})\|_{\bm{H}}^2 \\
&\overset{(\ref{c})}{=} \|\bm{g}^k - \tilde{\bm{g}}^k\|_{\bm{H}}^2 - \|(\bm{g}^k - \tilde{\bm{g}}^k) - \bm{M} (\bm{g}^k - \tilde{\bm{g}}^k)\|_{\bm{H}}^2 \\
&= 2 (\bm{g}^k - \tilde{\bm{g}}^k)^\top \bm{H} \bm{M} (\bm{g}^k - \tilde{\bm{g}}^k) - (\bm{g}^k - \tilde{\bm{g}}^k)^\top \bm{M}^\top \bm{H} \bm{M} (\bm{g}^k - \tilde{\bm{g}}^k) \\
&= (\bm{g}^k - \tilde{\bm{g}}^k)^\top (\bm{Q}^\top + \bm{Q} - \bm{M}^\top \bm{H} \bm{M}) (\bm{g}^k - \tilde{\bm{g}}^k) \\
&\overset{(\ref{m2})}{=} \|\bm{g}^k - \tilde{\bm{g}}^k\|_{\bm{G}}^2,
\end{aligned}
\end{equation}
where the second-to-last equality holds due to \(\bm{H} \bm{M} = \bm{Q}\) and \(2 \bm{g}^\top \bm{Q} \bm{g} = \bm{g}^\top (\bm{Q}^\top + \bm{Q}) \bm{g}\).

Substituting (\ref{pr2}) and (\ref{pr3}) into  (\ref{pr1}), we obtain:
\begin{equation}\label{pr4}
\begin{aligned}
\left[ \theta(\bm{f}) - \theta(\tilde{\bm{f}}^k) + (\bm{h} - \tilde{\bm{h}}^k)^\top \bm{F} (\tilde{\bm{h}}^k) \right] &\geq \frac{1}{2} \left( \|\bm{g} - \bm{g}^{k+1}\|_{\bm{H}}^2 - \|\bm{g} - \bm{g}^k\|_{\bm{H}}^2 \right) \\
&+ \frac{1}{2} \|\bm{g}^k - \tilde{\bm{g}}^k\|_{\bm{G}}^2, \quad \forall \bm{h} \in \Omega.
\end{aligned}
\end{equation}
Let \(\bm{g} = \bm{g}^*\) and $\bm{f} = \bm{f}^*$ in  (\ref{pr4}), which yields
\begin{equation}\label{pr5}
\begin{aligned}
&\|\bm{g}^k - \bm{g}^*\|_{\bm{H}}^2 - \|\bm{g}^{k+1} - \bm{g}^*\|_{\bm{H}}^2 \\
&\geq \|\bm{g}^k - \tilde{\bm{g}}^k\|_{\bm{G}}^2 + 2 \left[ \theta(\tilde{\bm{f}}^k) - \theta(\bm{f}^*) + (\tilde{\bm{h}}^k - \bm{h}^*)^\top \bm{F} (\tilde{\bm{h}}^k) \right].
\end{aligned}
\end{equation}
By the optimality of \(\bm{g}^*\) and the monotonicity of \(\bm{F}({\bm{h}})\) (see  (\ref{f})), we have:
\begin{equation}\label{pr6}
\begin{aligned}
&\theta(\tilde{\bm{f}}^k) - \theta(\bm{f}^*) + (\tilde{\bm{h}}^k - \bm{h}^*)^\top \bm{F} (\tilde{\bm{h}}^k) = \theta(\tilde{\bm{f}}^k) - \theta(\bm{f}^*) + (\tilde{\bm{h}}^k - \bm{h}^*)^\top \bm{F} ({\bm{h}^*}) \geq 0.
\end{aligned}
\end{equation}
The conclusion in  (\ref{res1}) follows directly from  (\ref{pr5}) and (\ref{pr6}).
\end{proof}

The contraction property mentioned above is crucial for the convergence of the sequence. The proof of sequence convergence derived from (\ref{res1}) has been extensively documented in the literature, including Theorem 2 in \cite{He2018A} and Theorem 4.1 in \cite{He2022A}. For completeness, this section provides a detailed proof process here.

\begin{thm}
For the sequence $\{ \bm{g}^k \}$ generated by the QPADM-slack(GB) algorithm (including Algorithms \ref{alg1} and \ref{alg3}), it will converge to some $\bm{g}^\infty$ belonging to $\mathcal{V}^*$, where $\mathcal{V}^* = \{ (\bm{b}^*, \bm{c}^*, \bm{d}^*) \mid (\bm{a}^*, \bm{b}^*, \bm{c}^*, \bm{d}^*) \in \Omega^* \}$.
\end{thm}

\begin{proof}
Based on (\ref{res1}), the sequence $\{ \bm{g}^k \}$ is bounded, and
\begin{align} \label{c1}
\lim_{k \to \infty} { \| \bm{g}^k - \tilde{\bm{g}}^k \| }_{\bm G} = 0.
\end{align}
Hence, $\{ \tilde{\bm{g}}^k \}$ is also bounded. Let $\bm{g}^\infty$ be an accumulation point of $\{ \tilde{\bm{g}}^k \}$ and $\{ \tilde{\bm{g}}^{k_j} \}$ be a subsequence converging to $\bm{g}^\infty$.

Recalling the statement in (\ref{vi2}), the sequences $\{ \tilde{\bm g}_k \}$ and $\{ \tilde{\bm g}_{k_j} \}$ are associated with $\{ \tilde{\bm h}_k \}$ and $\{ \tilde{\bm h}_{k_j} \}$ respectively. According to (\ref{p}), we have
\[
\theta(\bm f) - \theta( \tilde{ \bm f}^{k_j}) +  (\bm h - \tilde{ \bm h}^{k_j})^\top F ( \tilde{ \bm h}^{k_j}) \geq  \  (\bm g - \tilde{ \bm g}^{k_j})^\top \bm Q  (\bm g^k - \tilde{ \bm g}^{k_j}), \quad \forall \bm h \in \Omega.
\]
Noting that the matrix $\bm{G}$ is nonsingular (see  (\ref{m1})), this implies $\lim_{k \to \infty}(\bm{g}^k - \tilde{\bm{g}}^{k_j}) = 0$. Due to the continuity of $\theta(\bm{f})$ and $F(\bm{h})$, we obtain
\[
\bm{h}^\infty \in \Omega, \quad \theta(\bm{f}) - \theta(\bm{f}^\infty) + (\bm{h} - \bm{h}^\infty)^\top F(\bm{h}^\infty) \geq 0, \quad \forall \bm{h} \in \Omega.
\]
The above variational inequality indicates that $\bm{h}^\infty$ is a solution point of (\ref{vi}). Together with (\ref{res1}), we obtain  \begin{align}\label{eq22}
    \| \bm g^{k+1} -\bm g^\infty \|_H \leq \|\bm g^k -\bm g^\infty \|_H.
\end{align}

Furthermore, according to (\ref{c1}) and the fact that $\lim\limits_{j \to \infty} \tilde{\bm g}^{k_j} = \bm g^\infty$, the subsequence $\{\bm g^{k_j}\}$ also converges to $\bm g^\infty$. Then, in conjunction with (\ref{eq22}), it can be concluded that $\bm g^{k}$ does not possess more than one cluster point, thereby establishing the convergence of the sequence $\{\bm g^k\}$ to $\bm g^\infty$.
\end{proof}

\subsubsection{Linear Convergence Rate}
Here, we demonstrate that a worst-case convergence rate of \( \mathcal{O}\left({1}/{K}\right) \) in a non-ergodic sense can be established for QPADM-slack using Gaussian back substitution. To do this, we first need to prove the following proposition.
\begin{prop}\label{prop2}
For the sequence \(\{ \bm{g}^k \}\) generated by the QPADM-slack(GB) algorithm (which includes Algorithm \ref{alg1} and Algorithm \ref{alg3}), we have
\begin{align}
\| \bm{g}^{k+1} - \bm{g}^{k+2} \|_{\bm {H}}^2   \leq & \| \bm{g}^k - \bm{g}^{k+1} \|_{\bm {H}}^2, \label{res2}
\end{align}
where \( k > 0 \), and \( \bm{H} \) is defined in (\ref{m2}).
\end{prop}
\begin{proof}
First, by setting \( \bm h = \tilde{\bm h}^{k+1} \) in (\ref{p}), we obtain
\begin{align}\label{cr1}
\theta( \tilde{\bm f}^{k+1}) - \theta( \tilde{ \bm f}^k) +  (\tilde{\bm h}^{k+1}  - \tilde{ \bm h}^k)^\top F ( \tilde{ \bm h}^k) \geq & \  (\tilde{\bm g}^{k+1}  - \tilde{ \bm g}^k)^\top \bm Q  (\bm g^k - \tilde{ \bm g}^k). 
\end{align}
Note that (\ref{p}) is also true for $k = k + 1$. Thus, we also have
\begin{align*}
\theta(\bm f) - \theta( \tilde{ \bm f}^{k+1}) +  (\bm h - \tilde{ \bm h}^{k+1})^\top F ( \tilde{ \bm h}^{k+1}) \geq & \  (\bm g - \tilde{ \bm g}^{k+1})^\top \bm Q  (\bm g^{k+1} - \tilde{ \bm g}^{k+1}),  \ \forall \bm h \in \Omega. 
\end{align*}
Setting $\bm h = \tilde{\bm h}^k$ in the above inequality, we obtain
\begin{align}\label{cr2}
\theta(\tilde{ \bm f}^k ) - \theta( \tilde{ \bm f}^{k+1}) +  ( \tilde{ \bm h}^k - \tilde{ \bm h}^{k+1})^\top F ( \tilde{ \bm h}^{k+1}) \geq & \  ( \tilde{ \bm g}^k - \tilde{ \bm g}^{k+1})^\top \bm Q  (\bm g^{k+1} - \tilde{ \bm g}^{k+1}),  \ \forall \bm h \in \Omega. 
\end{align}
By adding (\ref{cr1}) and (\ref{cr2}) , and utilizing the antisymmetry of \( F \)  in (\ref{f}), we obtain
\begin{align}\label{cr3}
(\tilde{\bm g}^k - \tilde{\bm g}^{k+1})^\top \bm Q \left[ (\bm g^k - \tilde{\bm g}^k) - (\bm g^{k+1} - \tilde{\bm g}^{k+1}) \right] \geq 0.
\end{align}
By adding the term
\[
 \left[ (\bm g^k - \tilde{\bm g}^k) - (\bm g^{k+1} - \tilde{\bm g}^{k+1}) \right]^\top \bm Q \left[ (\bm g^k - \tilde{\bm g}^k) - (\bm g^{k+1} - \tilde{\bm g}^{k+1}) \right]
\]
to both sides of (\ref{cr3}), and using \( \bm g^\top \bm Q \bm g = \frac{1}{2} \bm g^\top (\bm Q^\top + \bm Q) \bm g \), we obtain
\[
(\bm g^k - \bm g^{k+1})^\top \bm Q  \left[ (\bm g^k - \tilde{\bm g}^k) - (\bm g^{k+1} - \tilde{\bm g}^{k+1}) \right] \geq \frac{1}{2} \left\|  (\bm g^k - \tilde{\bm g}^k) - (\bm g^{k+1} - \tilde{\bm g}^{k+1}) \right\|^2_{\bm Q^\top + \bm Q}.
\]
By substituting
\begin{align}\label{cc}
  (\bm{g}^k - \bm{g}^{k+1}) = \bm{M} (\bm{g}^k - \tilde{\bm{g}}^k)  \ \text{and} \  (\bm{g}^{k+1} - \bm{g}^{k+2}) = \bm{M} (\bm{g}^{k+1} - \tilde{\bm{g}}^{k+1})  
\end{align}
(from (\ref{c})) into the left-hand side of the last inequality, and using \( \bm{Q} \bm{M}^{-1} = \bm{H} \), it follows that
\begin{align}\label{cr4}
(\bm{g}^k - \bm{g}^{k+1})^\top \bm{H} \left[ (\bm{g}^k - \bm{g}^{k+1}) - (\bm{g}^{k+1} - \bm{g}^{k+2}) \right] \geq \frac{1}{2} \left\| (\bm{g}^k - \tilde{\bm{g}}^k) - (\bm{g}^{k+1} - \tilde{\bm{g}}^{k+1}) \right\|^2_{\bm{Q}^\top + \bm{Q}}.
\end{align}
Letting \( \bm a = (\bm g^k - {\bm g}^{k+1}) \) and \( \bm b = (\bm g^{k+1} - {\bm g}^{k+2}) \) in the identity (\ref{Id2}), we have
\begin{align*}
\| \bm g^k - {\bm g}^{k+1}\|_{\bm H}^2 - \| \bm g^{k+1} - {\bm g}^{k+2}\|_{\bm H}^2 &= 2 (\bm{g}^k - \bm{g}^{k+1})^\top \bm{H} \left[ (\bm{g}^k - \bm{g}^{k+1}) - (\bm{g}^{k+1} - \bm{g}^{k+2}) \right] \\
& -  \| (\bm{g}^k - \bm{g}^{k+1}) - (\bm{g}^{k+1} - \bm{g}^{k+2})  \|_{\bm H}^2. 
\end{align*}
By inserting (\ref{cr4}) into the first term on the right-hand side of the last equality, we obtain
\begin{align*}
\| \bm g^k - {\bm g}^{k+1}\|_{\bm H}^2 - \| \bm g^{k+1} - {\bm g}^{k+2}\|_{\bm H}^2 & = \left\| (\bm{g}^k - \tilde{\bm{g}}^k) - (\bm{g}^{k+1} - \tilde{\bm{g}}^{k+1}) \right\|^2_{\bm{Q}^\top + \bm{Q}} \\
 & -  \| (\bm{g}^k - \bm{g}^{k+1}) - (\bm{g}^{k+1} - \bm{g}^{k+2})  \|_{\bm H}^2. 
\end{align*}
By inserting (\ref{cc}) into the second term on the right-hand side of the last equality, we get 
\begin{align*}
\| \bm g^k - {\bm g}^{k+1}\|_{\bm H}^2 - \| \bm g^{k+1} - {\bm g}^{k+2}\|_{\bm H}^2 & = \left\| (\bm{g}^k - \tilde{\bm{g}}^k) - (\bm{g}^{k+1} - \tilde{\bm{g}}^{k+1}) \right\|^2_{\bm{Q}^\top + \bm{Q}} \\
 & -  \| \bm M (\bm{g}^k - \tilde{\bm{g}}^k)  - \bm M (\bm{g}^{k+1} - \tilde{\bm{g}}^{k+1})  \|_{\bm H}^2. 
\end{align*}
Since \( \bm{G} = (\bm{Q}^\top + \bm{Q}) - \bm{M}^\top \bm{H} \bm{M} \) is a positive-definite matrix, the assertion in (\ref{res2}) follows immediately.
\end{proof}

Note that \(\bm{Q}^\top + \bm{Q}\) and \(\bm{G} = (\bm{Q}^\top + \bm{Q}) - \bm{M}^\top \bm{H} \bm{M}\) are positive-definite matrices. Together with Proposition \ref{prop41}, there exists a constant \(c_0 > 0\) such that
\begin{align}\label{p21}
\|\bm{g}^{k+1} - \bm{g}^*\|_{\bm{H}}^2 \leq \|\bm{g}^k - \bm{g}^*\|_{\bm{H}}^2 - c_0 \|\bm{M}(\bm{g}^k - \tilde{\bm{g}}^k)\|_{\bm{H}}^2, \quad \forall \bm{g}^* \in \mathcal{V}^*.
\end{align}
 Proposition \ref{prop2} and (\ref{cc}) indicate that
\begin{align}\label{p22}
\| \bm{M} (\bm{g}^{k+1} - \tilde{\bm{g}}^{k+1}) \|_{\bm{H}}^2 \leq \| \bm{M} (\bm{g}^{k} - \tilde{\bm{g}}^{k}) \|_{\bm{H}}^2.
\end{align}
Now, with (\ref{p21}) and (\ref{p22}), we can establish the worst-case \( \mathcal{O}\left({1}/{K}\right) \)  convergence rate in a nonergodic sense for QPADM-slack(GB).

\begin{thm}
The sequence \(\{ \bm{g}^k \}\) generated by the QPADM-slack(GB) algorithm (including Algorithm \ref{alg1} and Algorithm \ref{alg3}) satisfies, for any positive integer \(K > 0\),
\begin{align}\label{p25}
\|\bm g^K - \bm g^{K+1} \|_{\bm{H}}^2 \leq \frac{1}{ c_0\left(K+1\right)}  \|\bm g^0 - \bm g^*\|_{\bm{H}}^2.
\end{align}
\end{thm}
\begin{proof}
First, it follows from (\ref{p21}) that
\begin{align}\label{p23}
\sum_{k=0}^{\infty} c_0 \| \bm{M} (\bm{g}^k - \tilde{\bm{g}}^k) \|_{\bm{H}}^2 \leq \| \bm{g}^0 - \bm{g}^* \|_{\bm{H}}^2, \quad \forall \bm{g}^* \in \mathcal V^*.
\end{align}
According to  (\ref{p22}), the sequence \(\{ \| \bm{M}(\bm{g}^k - \tilde{\bm{g}}^k) \|_{\bm{H}}^2 \}\) is monotonically non-increasing. Therefore, we have
\begin{align}\label{p24}
(K+1)\| \bm{M}(\bm{g}^{K+1} - \tilde{\bm{g}}^{K+1}) \|_{\bm{H}}^2 \leq \sum_{k=0}^{K} \| \bm{M}(\bm{g}^k - \tilde{\bm{g}}^k) \|_{\bm{H}}^2.
\end{align}
It  follows from(\ref{p23}) and (\ref{p24}) that \begin{align}\label{p26}
c_0 (K+1) \| \bm{M}(\bm{g}^{K+1} - \tilde{\bm{g}}^{K+1}) \|_{\bm{H}}^2  \leq \| \bm{g}^0 - \bm{g}^* \|_{\bm{H}}^2
\end{align}
The assertion (\ref{p25}) follows from (\ref{p26}) and $\bm M (\bm{g}^{K+1} - \tilde{\bm{g}}^{K+1}) = (\bm g^K - \bm g^{K+1})$ immediately.
\end{proof}

\section{Supplementary Experiments}\label{4A}
Here, we will supplement numerical experiments on nonconvex regularization regression and classification problems.
Note that when employing algorithms QPADM-slack(GB) and M-QPADM-slack(GB) to solve the problem of nonconvex regularizations, the LLA method (referred to as Algorithm \ref{alg2}) is actually used. Consequently, the number of iterations (Ite) corresponds to the total number of iterations generated within Algorithm \ref{alg2}. Typically, Algorithm \ref{alg2} involves two ADMM  iteration processes: one for solving the unweighted $\ell_1$ regularization problem and another for solving the weighted $\ell_1$ regularization problem.  This counting method was also adopted by \cite{Wen2024Nonconvex} in their ADMM algorithms.
We have implemented warm-start technique (\cite{Friedman2010Regularization}), which allows the second ADMM algorithm to converge rapidly, typically requiring only around ten to twenty iterations in most cases.

\subsection{Supplementary Experiments for Section \ref{sec451}}\label{4A1}
In this section, we first supplement the experiments with nonconvex regularization terms for Table \ref{tab41} (non-parallel computing environment), where the dimensions are chosen as $(n, p) = (1000, 30000)$ and $(30000, 1000)$. The specific numerical results are provided in Table \ref{tab44}.  
The numerical results indicate that when solving the quantile regression problem with MCP and SCAD penalties, the ADMM algorithm with Gaussian back substitution still outperforms the one without it. Additionally, it is observed that the number of iterations required by algorithms QPADM-slack(GB) and M-QPADM-slack(GB) for these two nonconvex penalized quantile regressions is no more than 20 steps greater than that required for solving the quantile regression with $\ell_1$ penalty. This advantage is attributed to the warm-start technique, where the solution obtained from the $\ell_1$ penalized regularization serves as the initial solution for the weighted $\ell_1$ penalty regularization.
\begin{table}[H]
\caption{Comparison results of ADMM algorithms under  nonconvex penalties.}\label{tab44}
\renewcommand{\arraystretch}{1.5}
\resizebox{1\linewidth}{!}{
\begin{threeparttable}
\begin{tabular}{lcccccccccc}
\hline
MCP &  \multicolumn{5}{c}{$(n,p) = (30000,1000)$}  &\multicolumn{5}{c}{$(n,p) =(1000,30000)$}  \\ 
\cmidrule(lr){2-6}\cmidrule(lr){7-11}
 &  P1 &  P2 &  AE   &  Ite   &  Time  &  P1 &  P2 &  AE   &  Ite   &  Time \\ \hline
QPADM        & 100 & 100 & \textbf{0.014(0.03)} & 138.1(14.5)& 41.35(6.78) & 82 & 100 & 7.532(1.12) & 500+(0.0)  & 2038.7(68.6)  \\
QPADMslack   & 100 & 100 & 0.051(0.05) & 68.3(6.69)& 37.56(5.27) & 89 & 100 & 7.807(1.33) & 266(28.6) & 1276.5(30.6)  \\
QPADM-slack(GB)      & 100 & 100 & 0.025(0.04) & 52.1(5.44)  & 35.17(5.30) & 93 & 100 & 7.352(1.08) & 241(22.8)  & 1101.4(27.6) \\
M-QPADM-slack(GB)       & 100 & 100 & 0.017(0.03) & \textbf{39.5(4.78)} & \textbf{24.42(3.08)} & \textbf{96} & 100 & \textbf{7.082(0.92)} & \textbf{172(16.3)} & \textbf{885.7(19.4)} \\
\midrule[1pt]
 SCAD & \multicolumn{5}{c}{$(n,p) = (30000,1000)$} & \multicolumn{5}{c}{$(n,p) =(1000,30000)$} \\ 
\cmidrule(lr){2-6}\cmidrule(lr){7-11}
 &  P1 &  P2 &  AE   &  Ite     &  Time  &  P1 &  P2 &  AE   &  Ite     &  Time\\ \hline
QPADM  & 100 & 100 & \textbf{0.012(0.03)} & 142.3(14.9)  & 44.2(6.91) &85 & 100 & 7.233(1.02) & 500+(0.0)   & 2056.2(67.9)\\
QPADMslack  & 100 & 100 & 0.056(0.05) & 65.2(6.22)  & 34.78(5.10)  &88 & 100 & 7.886(1.38) & 251(24.6)  & 1202.5(28.9) \\
QPADM-slack(GB)   &100& 100 & 0.026(0.04) & 53.0(5.61)  & 37.09(5.5) &94 & 100 & 7.426(1.12) & 261(25.7)   & 1208.6(29.6)\\
M-QPADM-slack(GB)   &{100}& 100 & {0.015(0.03)} & \textbf{40.9(4.91)} & \textbf{26.32(3.17)} &\textbf{98} & 100 & \textbf{7.021(0.89)} &\textbf{ 175(16.7)}  & \textbf{892.2(20.2)} \\ \hline
\end{tabular}
\begin{tablenotes}
        \footnotesize
        \item[*] The symbols in this table are defined as follows: P1 (\%) represents the proportion of times $x_1$ is selected; P2 (\%) represents the proportion of times $x_6$, $x_{12}$, $x_{15}$, and $x_{20}$ are selected; AE denotes the absolute estimation error; Ite indicates the number of iterations; and Time (s) refers to the running time. The numbers in parentheses represent the corresponding standard deviations, and the optimal results are highlighted in bold.
\end{tablenotes}
\end{threeparttable}}
\end{table}

Next,  we  supplement the experiments with nonconvex regularizer terms for Table \ref{tab42} (parallel computing environment), where the dimensions are chosen as $(n, p) = (500000, 1000)$. The specific numerical results are provided in Table \ref{tab45}, which demonstrate that the Gaussian back substitution technique effectively enhances both the accuracy and efficiency of the QPADM-slack algorithm.
\begin{table}[H]
\caption{Comparison  results of different P-ADMM algorithms under nonconvex penalties.}\label{tab45}
\renewcommand{\arraystretch}{1.5}
\resizebox{1\columnwidth}{!}{
\begin{threeparttable}
\begin{tabular}{lcccccccc}
\hline
  & \multicolumn{2}{l}{QPADM} & \multicolumn{2}{r}{MCP} & \multicolumn{2}{l}{QPADM} & \multicolumn{2}{r}{SCAD}\\ 
\cmidrule(lr){2-5}\cmidrule(lr){6-9}
M & Nonzero &  AE   &  Ite   &  Time  & Nonzero &  AE   &  Ite   &  Time   \\ \hline
 5  & 26.5(2.01) & {0.034(0.0005)} & 451.3(37.5) & 180.4(15.02) & 27.1(1.99) & 0.035(0.0005) & 479.2(39.2) & 189.3(14.3)  \\
10  & 27.2(2.09) & 0.039(0.0006) & 475.7(40.1) & 98.3(6.99) & 28.5(2.07) & 0.040(0.0006) & 496.3(42.7) & 99.2(7.12) \\
20 & 29.7(2.16) & 0.045(0.0008) & 500+(0.00) & 49.8(4.11) & 29.9(2.18) & 0.047(0.0007) & 500+(0.00) & 50.2(4.03)  \\\hline
   & \multicolumn{2}{l}{QPADM-slack}&\multicolumn{2}{r}{MCP} & \multicolumn{2}{l}{QPADM-slack} & \multicolumn{2}{r}{SCAD} \\  \cmidrule(lr){2-5}\cmidrule(lr){6-9}
M  &  Nonzero &  AE           &  Ite        &  Time & Nonzero & AE & Ite & Time \\\hline
 5  & 25.1(1.97) & 0.037(0.0006) & 332.5(30.0) & 141.2(9.66) & 24.7(1.95) & 0.038(0.0006) & 341.6(29.8) & 136.4(9.52)  \\
 10 & 29.1(2.06) & 0.044(0.0007) & 360.1(34.5) & 75.2(5.23) & 28.9(2.11) & 0.043(0.0007) & 366.8(32.3) & 76.6(5.15)   \\
20 & 31.1(2.17) & 0.050(0.0008) & 431.1(41.7) & {42.1(2.52)} & 30.1(2.23) & 0.049(0.0008) & 425.9(40.2) & 43.2(2.74)  \\\hline
   & \multicolumn{2}{l}{QPADM-slack(GB) } & \multicolumn{2}{r}{MCP} & \multicolumn{2}{l}{QPADM-slack(GB)} & \multicolumn{2}{r}{SCAD}\\   \cmidrule(lr){2-5}\cmidrule(lr){6-9}
M  &  Nonzero &  AE           &  Ite        &  Time  & Nonzero & AE & Ite & Time \\\hline
 5  & {23.5(1.91)} & 0.028(0.0005) & {290.1(24.8)}   & {87.2(6.53)} & {23.1(1.88)} & {0.031(0.0005)} & {271.2(25.2)} & {88.9(6.81)}    \\
 10 & {24.4(1.98)} & {0.032(0.0005)} &{263.0(22.4)}  & {55.1(3.99)}  & {24.2(1.91)} & {0.035(0.0005)} & {269.3(23.5)} & {54.6(4.02)}  \\
20 & {25.1(2.03)} & {0.035(0.0006)} & {291.2(27.0)}   & {39.5(2.35)}  & {25.3(2.00)} & {0.038(0.0006)} & {283.4(27.6)} & {37.3(2.21)}  \\\hline
 & \multicolumn{2}{l}{M-QPADM-slack(GB)  } & \multicolumn{2}{r}{MCP} & \multicolumn{2}{l}{M-QPADM-slack(GB)} & \multicolumn{2}{r}{SCAD}\\  \cmidrule(lr){2-5}\cmidrule(lr){6-9}
M  &  Nonzero &  AE           &  Ite        &  Time  & Nonzero & AE & Ite & Time  \\\hline
 5  & \bf{14.10(1.40)} & \bf{0.020(0.0004)} & \bf{228.1(14.9)}   & \bf{71.1(4.90)}  & \bf{14.05(1.37)} & \bf{0.022(0.0004)} & \bf{215.3(15.1)} & \bf{70.5(4.84)}  \\
 10 & \bf{14.5(1.47)} & \bf{0.028(0.0005)} &\bf{232.2(17.83)}  & \bf{40.1(2.73)}  & \bf{14.77(1.48)} & \bf{0.027(0.0005)} & \bf{219.6(16.2)} & \bf{39.2(2.77)}  \\
20 & \bf{15.3(1.51)} & \bf{0.034(0.0005)} & \bf{230.7(13.8)}   & \bf{22.0(1.55)}  & \bf{15.01(1.54)} & \bf{0.031(0.0005)} & \bf{221.3(14.1)} & \bf{21.3(1.59)}  \\\hline
\end{tabular}
\begin{tablenotes}
        \footnotesize
        \item[*] Since the values of all methods for metrics P1 and P2 are 100, they are not listed in Table \ref{tab45}. ``Nonzero" indicates the number of non-zero coefficients in the estimates. The numbers in parentheses represent the corresponding standard deviations, and the optimal results are shown in bold.
\end{tablenotes}
\end{threeparttable}}
\end{table}

\subsection{Supplementary Experiments for Section \ref{sec452} }\label{4A2}
In this section, we provide additional experiments involving nonconvex regularization terms for the data presented in Table \ref{tab46} and Table \ref{tab47}. All experimental settings remain consistent with those described in Section \ref{sec452}. The numerical results in both tables demonstrate that incorporating Gaussian back substitution steps into the QPADM algorithm significantly enhances the sparsity and accuracy of nonconvex SVM classifiers.

\begin{table}[!ht]\footnotesize
    \centering
    \renewcommand{\arraystretch}{1.5}
    \caption{Comparative analysis of Sparsity, training, and average testing accuracies (\%) for SCAD-SVM using M-QPADM-slack(GB), QPADM-slack, and QPADM-slack(GB)}
    \begin{tabular}{llllllllll}
    \hline
          & \multicolumn{3}{c}{M-QPADM-slack(GB)} & \multicolumn{3}{c}{QPADM-slack(GB)} & \multicolumn{3}{c}{QPADM-slack}\\ 
        \cmidrule(lr){2-4}\cmidrule(lr){5-7}\cmidrule(lr){8-10}
        M & Sparsity & Train & Test & Sparsity & Train & Test & Sparsity & Train & Test \\ \hline
        2 & \bf{92.00} & \bf{99.80} & \bf{97.80} & 87.50 & 95.50 & 93.50 & 85.00 & 94.00 & 92.50 \\ 
        4 & \bf{92.00} & \bf{99.60} & \bf{97.50} & 85.50 & 94.80 & 92.50 & 82.00 & 93.00 & 91.50 \\ 
        6 & \bf{92.00} & \bf{99.50} & \bf{97.40} & 84.50 & 94.00 & 91.60 & 81.00 & 92.00 & 90.60 \\ 
        8 & \bf{91.80} & \bf{99.20} & \bf{97.10} & 83.00 & 93.20 & 91.00 & 80.00 & 91.20 & 90.00 \\ 
        10 & \bf{91.60} & \bf{99.00} & \bf{96.90} & 82.00 & 92.80 & 90.20 & 79.00 & 90.80 & 89.20 \\ 
        12 & \bf{91.50} & \bf{98.80} & \bf{96.70} & 81.50 & 92.00 & 89.50 & 77.50 & 90.00 & 88.50 \\ 
        14 & \bf{91.40} & \bf{98.70} & \bf{96.60} & 80.50 & 91.30 & 88.70 & 76.50 & 89.80 & 87.80 \\ 
        16 & \bf{91.20} & \bf{98.50} & \bf{96.40} & 79.50 & 90.50 & 87.50 & 74.50 & 88.50 & 87.20 \\ 
        18 & \bf{91.10} & \bf{98.40} & \bf{96.30} & 78.50 & 89.80 & 86.50 & 73.50 & 87.80 & 86.30 \\ 
        20 & \bf{91.00} & \bf{98.30} & \bf{96.20} & 77.50 & 88.80 & 85.50 & 72.50 & 87.70 & 85.20 \\ \hline
    \end{tabular}
    \label{tab46}
\end{table}

\begin{table}[!ht]\footnotesize
    \centering
    \renewcommand{\arraystretch}{1.5}
    \caption{Comparative analysis of Sparsity, training, and average testing accuracies (\%) for MCP-SVM using M-QPADM-slack(GB), QPADM-slack, and QPADM-slack(GB)}
    \begin{tabular}{llllllllll}
    \hline
          & \multicolumn{3}{c}{{M-QPADM-slack(GB)}} & \multicolumn{3}{c}{{QPADM-slack(GB)}} & \multicolumn{3}{c}{{QPADM-slack}}\\ 
        \cmidrule(lr){2-4}\cmidrule(lr){5-7}\cmidrule(lr){8-10}
       {M} & {Sparsity} & {Train} & {Test} &{Sparsity} & {Train} & {Test} & {Sparsity} & {Train} & {Test} \\ \hline
        2 & \bf{92.20} & \bf{99.90} & \bf{97.90} & 87.30 & 95.40 & 93.40 & 84.80 & 93.80 & 92.40 \\ 
        4 & \bf{92.10} & \bf{99.70} & \bf{97.60} & 85.30 & 94.70 & 92.40 & 81.80 & 92.90 & 91.40 \\ 
        6 & \bf{92.00} & \bf{99.60} & \bf{97.50} & 84.30 & 93.90 & 91.50 & 80.80 & 91.90 & 90.50 \\ 
        8 & \bf{91.90} & \bf{99.30} & \bf{97.20} & 82.80 & 93.10 & 90.90 & 79.80 & 91.10 & 89.90 \\ 
        10 & \bf{91.70} & \bf{99.10} & \bf{97.00} & 81.80 & 92.70 & 90.10 & 78.80 & 90.70 & 89.10 \\ 
        12 & \bf{91.60} & \bf{98.90} & \bf{96.80} & 81.30 & 91.90 & 89.40 & 77.30 & 89.90 & 88.40 \\ 
        14 & \bf{91.50} & \bf{98.80} & \bf{96.70} & 80.30 & 91.20 & 88.60 & 76.30 & 89.70 & 87.70 \\ 
        16 & \bf{91.30} & \bf{98.60} & \bf{96.50} & 79.30 & 90.40 & 87.40 & 74.30 & 88.40 & 87.10 \\ 
        18 & \bf{91.20} & \bf{98.50} & \bf{96.40} & 78.30 & 89.70 & 86.40 & 73.30 & 87.70 & 86.20 \\ 
        20 & \bf{91.10} & \bf{98.40} & \bf{96.30} & 77.30 & 88.70 & 85.40 & 72.30 & 87.60 & 85.10 \\ \hline
    \end{tabular}
    \label{tab47}
\end{table}

\end{document}